\documentclass[12pt,draftclsnofoot,letterpaper,onecolumn,romanappendices]{IEEEtran}
\usepackage[dvips]{graphicx}
\usepackage{cite}
\usepackage{epsfig}
\usepackage{amsthm}
\usepackage{amsmath,amssymb,amsfonts,amstext,amsbsy,amsopn,dsfont}
\usepackage{cases}
\usepackage{sublabel}
\usepackage{array}

\newtheorem{theorem}{\textbf{Theorem}}
\newtheorem{proposition}{\textbf{Proposition}}
\newtheorem{lemma}{\textbf{Lemma}}
\newtheorem{remark}{\textbf{Remark}}
\newtheorem{definition}{\textbf{Definition}}

\newcommand{\defn}{\triangleq}
\newcommand{\dif}{\textmd{d}}
\newcommand{\ie}{i.e., }

\begin{document}


\title{Multicast Outage Probability and Transmission Capacity of Multihop Wireless Networks}

\author{Chun-Hung Liu and Jeffrey G. Andrews
\thanks{C.-H. Liu and J. G. Andrews are with the Department of Electrical and Computer Engineering, the University of Texas at Austin, Austin TX 78712-0204, USA. The contact author is J. G. Andrews (Email: jandrews@ece.utexas.edu). Manuscript date: \today.}}

\maketitle

\begin{abstract}
Multicast transmission, wherein the same packet must be delivered to multiple receivers, is an important aspect of sensor and tactical networks and has several distinctive traits as opposed to more commonly studied unicast networks.  Specially, these include (i) identical packets must be delivered successfully to several nodes, (ii) outage at any receiver requires the packet to be retransmitted at least to that receiver, and (iii) the multicast rate is dominated by the receiver with the weakest link in order to minimize outage and retransmission. A first contribution of this paper is the development of a tractable multicast model and throughput metric that captures each of these key traits in a multicast wireless network. We utilize a Poisson cluster process (PCP) consisting of a distinct Poisson point process (PPP) for the transmitters and receivers, and then define the multicast transmission capacity (MTC) as the maximum achievable multicast rate per transmission attempt times the maximum intensity of multicast clusters under decoding delay and multicast outage constraints. A multicast cluster is a contiguous area over which a packet is multicasted, and  to reduce outage it can be tessellated into $v$ smaller regions of multicast. The second contribution of the paper is the analysis of several key aspects of this model, for which we develop the following main result. Assuming $\tau/v$ transmission attempts are allowed for each tessellated region in a multicast cluster, we show that the MTC is $\Theta\left(\rho k^{x}\log(k)v^{y}\right)$ where $\rho$, $x$ and $y$ are functions of $\tau$ and $v$ depending on the network size and intensity, and $k$ is the average number of the intended receivers in a cluster.  We derive $\{\rho, x, y\}$ for a number of regimes of interest, and also show that an appropriate number of retransmissions can significantly enhance the MTC.
\end{abstract}

\begin{keywords}
Multicast Transmission, Multicast Outage, Network Capacity, Information Theory, Stochastic Geometry.
\end{keywords}

\section{Introduction}

Multicast refers to the scenario whereby a transmitter needs to send a packet to multiple receivers.  In a wireless network, this creates a two-edged sword. On one hand, the broadcast nature of wireless transmission assists multicast; but roughly uncorrelated outage probabilities at each receiver (due to spatially distinct fading and interference) require retransmissions that cause interference and waste.  Multicast is an important aspect of sensor and tactical networks, and increasingly in commercial networks where streaming is supported.  However, the literature on multicast is minuscule compared to unicast -- whereby nodes are paired into sources and destinations -- and even basic modeling issues such as outage, capacity/throughput definitions, and retransmissions are not widely agreed upon. In this work, we attempt to investigate the fundamental throughput limits of multicast transmission and we develop a metric based on spatial outage capacity which we term \emph{multicast transmission capacity} (MTC).

In order to characterize the MTC in a wireless network we propose a multicast network model in which each transmitter has an intended multicast region (called a cluster) where all the intended receivers are uniformly and independently scattered, and hence a Poisson cluster process can be reasonably used to model the transmit-receiver location statistics. The active transmitters are modeled as a stationary Poisson point process (PPP) and their associated receiver nodes in the cluster are also a stationary PPP, as shown in Fig. \ref{Fig:MulticastModel}. In other words, each cluster is randomly located in the network and comprises a multicast session.  This paper will develop interference and outage expressions for this model, and analyze some important cases of the network model and design space, including the network intensity, size and the effect of retransmissions.

\subsection{Motivation and Related Work}

The majority of the existing works on network capacity are focused on the \emph{unicast} scenario and built upon the protocol and physical network models proposed in \cite{PGPRK00}. Generally speaking, the unicast network capacity is the maximum number of point-to-point communication links that can be simultaneously supported in the network under some transmission constraints. For example, transport capacity built upon the protocol model \cite{PGPRK00} has a geometric constraint on transmitter-receiver pairs. Due to the difficulty in coordinating the transmission constraints between multiple receivers, the unicast capacity, in general, is not readily extended to the multicast capacity. Even the definition of multicast capacity is not widely agreed upon.

Some previous works, such as \cite{SSXLRS07,UNPGDS09,XYL09,AKVRRR06,BTAVIL06,PCSS05,PJGR05,AKHVRRR06,RZHENG06,AKHRR07}, have made significant progress in studying the multicast or broadcast capacity\footnote{Broadcast means only one transmitter would like to transmit to all of other nodes in the network, whereas multicast means transmitters transmit to a certain number of nodes in the network. Hence, broadcast capacity can be viewed as a special case of multicast capacity.}. For example, in \cite{SSXLRS07} the protocol model is used where source nodes and their multicast destinations are randomly chosen. The multicast capacity is defined as the sum rate of all multicast flows and it is obtained as a function of the number of multicast sources. In \cite{XYL09}, the multicast capacity under the protocol model is defined as the transmission rate summed over all of the multicast traffic flows in the network. Its scaling characterization is obtained by the number of receivers in each multicast session. Reference \cite{AKHVRRR06} showed that the broadcast capacity under the protocol model does not change by more than a constant factor when the number of nodes, the radio range or the area of the network is changed. In \cite{RZHENG06}, the physical model and a stationary PPP of the nodes in the network are considered. It showed that the broadcast capacity is a constant factor of the computed upper bound when the number of nodes goes to infinity under a constant node intensity. The multicast and broadcast capacities in the above works are defined behind the main concept that the transmitted information should be received by all of its intended receivers. However, they are not investigated from the multi-receiver outage point of view and thus their scaling results cannot provide us retransmission guidelines for capacity enhancement.

The multicast capacity problem in this paper is studied from an outage perspective, which is a departure from previous work. For example, the prior work built on the protocol model in \cite{PGPRK00} does not consider channel impairments such as fading and path loss, and thus, all transmissions are successful once certain geometric constraints are satisfied. Also, outage and the resulting retransmissions are not considered. This is somewhat unrealistic. The main issues we would like to clarify are (i) how many simultaneous multicast sessions can (and should) coexist when all receivers in a multicast session need to receive the packet from their transmitter, and (ii) when are retransmissions beneficial or detrimental to the multicast capacity? Hence, we introduce the MTC in a planar network, which is defined as the maximum achievable multicast rate per transmission attempt times the maximum number of the coexisting clusters in the network per unit area, subject to decoding delay and multicast outage constraints. This is a logical evolution of the transmission capacity framework originated in \cite{SWXYJGAGDV05} to the topic of multicast. The decoding delay constraint here means a transmitter can multicast a packet to all of its intended receivers up to $\tau\in\mathbb{N}_+$ transmission attempts and multicast outage happens when any of the intended receivers in a cluster does not receive the information multicasted by their transmitter during those $\tau$ attempts.

\subsection{Main Contributions}
Since unicast outage cannot be directly applied to the outage scenario of multiple receivers, our first contribution is introducing MTC with multicast outage. To the best of our knowledge, this is the first work to study multicast capacity under a multicast outage constraint. According to the proposed Poisson cluster model for multicast, we also propose some new cluster-based definitions for the largeness and denseness of a network in order to characterize the scaling behaviors of the MTC under different conditions. Referring to Table \ref{Tab:MainVars} for the notation of main variables, we have shown that the scaling of both single-hop and multihop MTC can be expressed in a general form of $\Theta(\rho k^{x}\log(k)v^{y})$ where $\rho$, $x$ and $y$ are given in Table \ref{Tab:MainResultsMTC} for different network conditions, $k$ is the average number of the intended receivers in a cluster, and $v$ is the number of the tessellated regions of equal area in a cluster\footnote{Throughout this paper, standard asymptotic notation will be used: $O(\cdot)$, $\Omega(\cdot)$ and $\Theta(\cdot)$ correspond to asymptotic upper, lower, and tight bounds, respectively.}. From Table \ref{Tab:MainResultsMTC}, we know retransmissions have a significant effect on the MTC and certain number of retransmissions could enhance it. We found that the decoding delay $\tau$ can be viewed as a resource which should be allocated properly in every tesselated region to maximize the MTC. In addition, we also show that the MTC scaling holds for the various parameters of Nakagami fading.

We characterize three approaches to improving the MTC -- interference-suppression, interference-avoidance and area-shrinking (i.e. reducing the cluster size). The area-shrinking method is able to provide the best capacity gain among the three since the multicast outage probability is reduced due to fewer receivers in a cluster. However, shrinking the area of a cluster means some of the intended receivers may have to be excluded, which is not welcome. Thus the multihop multicast method is proposed to improve the MTC without shrinking the cluster. Our main result shows that if the clusters are appropriately tessellated multihop multicast can significantly improve the MTC compared to its single-hop counterpart. This is because an appropriate number of retransmissions largely reduces the multicast outage probability and thus results in the increase of the maximum contention intensity that compensates the loss of spectral efficiency due to retransmissions.

\subsection{Paper Organization}
In Section \ref{Sec:MTCmodelPrelims}, the network model and its assumptions are described and some preliminaries for the following analysis are provided here. The main results of multicast  transmission capacity for the case of single-hop multicast are presented in Section \ref{Sec:MTCwSinglehop}, whereas Section \ref{Sec:MTCwMultihop} has the main results of multicast transmission capacity with multihop multicast. Finally, we conclude our work in Section \ref{Sec:Conclusion}.


\begin{table}[!h]
  \centering
  \caption{Notation of Main Variables, Processes and Functions}\label{Tab:MainVars}
  \begin{tabular}{|c|c|}
  \hline
  Symbol & Definition\\ \hline
  $\Phi^{\textsf{t}}\,(\Phi^{\textsf{r}},\,\Phi^{\textsf{c}})$  &  PPP of transmitters (receivers, connected receivers)\\
  $\epsilon$ & upper bound of multicast outage probability ($\ll 1$)\\
  $\lambda_{\textsf{t}}\,(\lambda_{\textsf{r}},\,\lambda_{\textsf{c}})$ & intensity of $\Phi^{\textsf{t}}\,(\Phi^{\textsf{r}},\,\Phi^{\textsf{c}})$\\
  $\bar{\lambda}_{\textsf{t}}$ & maximum contention intensity of $\lambda_{\textsf{t}}$ \\
  $\bar{\lambda}_{\epsilon}$ & first-order Taylor expansion of $\bar{\lambda}_{\textsf{t}}$ ($\bar{\lambda}_{\epsilon}\approx \bar{\lambda}_{\textsf{t}}$ for small $\epsilon$)\\
  $\mathcal{E}(\lambda_{\textsf{t}})$ & multicast outage event depending on $\lambda_{\textsf{t}}$\\
  $\tau$ & decoding delay constraint, positive integer \\
  $v$ & number of the tessellated regions in a cluster ($v\leq \tau$)\\
  $s$ & radius of a multicast cluster ($s>1$)\\
  $\mu(\mathcal{A})$ & Lebesgue measure of a bounded set $\mathcal{A}$\\
  $\beta$ & SIR threshold \\
  $\alpha$ & path loss exponent ($\alpha>2$)\\
  $b$ & multicast transmission rate (bps/Hz)\\
  $C_{\epsilon}$ & multicast transmission capacity (bps/Hz/$m^2$)\\
  $k$ &  average number of the intended receivers in a cluster, $\pi s^2\lambda_{\textsf{r}}$\\
  $f_H(\cdot), F_H(\cdot), F^{\texttt{c}}_H(\cdot)$ & PDF, CDF, CCDF of channel fading gain $H$\\
   \hline
  \end{tabular}
\end{table}


\begin{table}[!h]
\centering
\caption{Summary of Main Results on Multicast Transmission Capacity}\label{Tab:MainResultsMTC}
\begin{tabular}{|c|c|c|c|}
\hline
\bf Multicast Transmission Capacity $C_{\epsilon}$& \multicolumn{3}{c|} {$\Theta\left(\rho\,k^{x}\log(k)\,v^{y}\right)$}\\
\hline
\bf Maximum Contention Intensity $\bar{\lambda}_{\epsilon}$ &  \multicolumn{3}{c|} {$\Theta(\rho\,k^{x}\,v^y)$} \\
\hline
\bf Network Condition & Dense & Large  & Large Dense \\
\hline
\bf $x$ & $-\frac{v}{\tau}$ & $-\left(1+\frac{v}{\tau}\right)$ & $-\left(\frac{\tau+2v}{2\tau}\right)$\\
\hline
$\rho$ &  \multicolumn{3}{c|} {$\frac{1}{\tau^2}(\epsilon(\tau/v+1))^{\frac{v}{\tau}}$} \\
\hline
$y$ & \multicolumn{3}{c|} {$\frac{v}{\tau}+1$}\\
\hline
\end{tabular}
\end{table}

\section{Network Model and Preliminaries}\label{Sec:MTCmodelPrelims}

In most of the existing literature, multicast capacity is defined based on the sum of the supportable rates of all multicast sessions. Similarly, the MTC here is characterized by the sum of the multicast rates of the multicast sessions not in (multicast) outage. The multicast rate is affected by the number of the receivers as well as the interference from other multicast sessions. The multicast network model is developed with the principle that it should capture the key traits of a multicast network (for example spatial reuse, retransmissions, broadcast packets, etc.), while being as tractable as possible.

\subsection{Clustered Network Model for Multicast Transmission}\label{Sec:MultiTranModelRxConnProc}
In the network, each transmitter has a multicast cluster of equal area and its receive nodes in the cluster suffer aggregate interference from a Poisson field of transmitters. Specifically, we assume that the network is operating a slotted ALOHA protocol and the distribution of the transmitting nodes in the network is a stationary Poisson point process (PPP) $\Phi^\textsf{t}$ of intensity $\lambda_{\textsf{t}}$. As shown in Fig. \ref{Fig:MulticastModel}, any transmitter $X_i\in\Phi^{\textsf{t}}$ has its own intended multicast cluster $\mathcal{R}_i$ where all of its intended receivers are uniformly and independently distributed and they also form a stationary PPP $\Phi^{\textsf{r}}_i$ of intensity $\lambda_{\textsf{r}}$. \emph{Note that each cluster could contain other transmitters and unintended receivers in addition to its own transmitter and intended receivers.}

Accordingly, the multicast transmission sessions in the network follow a the Poisson cluster process (PCP) $\mathcal{Z}_i\defn \Phi^{\textsf{r}}_i\cup X_i$, \ie each transmitter is a \emph{parent} node associated with a cluster of receive \emph{daughter} nodes. The cluster processes $\{\mathcal{Z}_{i}\}$ corresponding to different transmitters $\{X_i\}$ are assumed to be independent so that the superposition of all clusters yields the resulting cluster process $\Phi=\bigcup_{X_i\in\Phi^{\textsf{t}}}\mathcal{Z}_i$ of intensity $\lambda=k\,\lambda_{\textsf{t}}$ where $k=\pi s^2\lambda_{\textsf{r}}$ is the average number of the intended receivers in each cluster assuming all $\{\mathcal{R}_i, \forall i\in \mathbb{N}\}$ have the same radius $s>1$. The distribution of the intended receiver nodes in each cluster is modeled as a \emph{marked} PPP denoted by $\Phi^{\textsf{r}}_i \defn \{(Y_{ij},H_{ij}): Y_{ij}\in \mathcal{R}_i,j\in\mathbb{N}\}$, where $H_{ij}$ is the fading channel gain between transmitter $X_i$ and its intended receiver $Y_{ij}$. Similarly, the distribution of transmitters in the network is also a marked PPP, \ie $\Phi^{\textsf{t}} \defn \{(X_i,\{\tilde{H}_{ij}\}), i,j\in \mathbb{N}_+\}$ where $\tilde{H}_{ij}$ denotes the fading channel gain between transmitter $X_i$ and the receiver $Y_{0j}$ located in cluster $\mathcal{R}_0$.\footnote{Since all of the following analysis is based on the nodes in the reference cluster $\mathcal{R}_0$, the subscript 0 of some variables will not be explicitly indicated if there is no ambiguity. So  $Y_j$ and $H_j$ in $\mathcal{R}_0$ actually stand for $Y_{0j}$ and $H_{0j}$, respectively.} All fading channel gains are i.i.d. with  probability density function (PDF) $f_{H}(\cdot)$.

Without loss of generality, the MTC can be evaluated in the reference cluster $\mathcal{R}_0$ whose transmitter $X_0$ is located at the origin. We condition on this typical transmitter $X_0$ resulting in what is known the Palm distribution for transmitting nodes in the two-dimensional Euclidean space \cite{DSWKJM96}. It follows by Slivnyak's theorem \cite{DSWKJM96} that this conditional distribution also corresponds to a homogenous PPP with the same intensity and an additional point at the origin. The signal propagation in space is assumed to undergo path loss and fading. The path loss model between two nodes $X$ and $Y$ used in this paper is
\begin{equation}\label{Eqn:PathLossModel}
\ell(|X-Y|) \defn\begin{cases}|X-Y|^{-\alpha},&\quad \text{if}\,\, |X-Y|\geq 1 \\ 0, &\quad \text{else}, \end{cases}
\end{equation}
where $|X-Y|$ denotes the Euclidean distance between nodes $X$ and $Y$, and $\alpha>2$ is the path loss exponent\footnote{In a planar network, $\alpha$ is greater than 2 in order to have bounded interference, i.e. $I_j<\infty$ almost surely if $\alpha>2$ \cite{FBBBPM06,MHJGAFBODMF10}}.
The reason of using the model in \eqref{Eqn:PathLossModel} is because the model $|\cdot|^{-\alpha}$ does not behave well in the near field of each transmitter and it thus leads to an unbounded mean of the shot noise process. This model is similar to the idea of the bounded propagation model proposed in \cite{OAZJH04}\cite{ODPT04}. The Nakagami-$m$ fading model is adopted in this paper because it covers several different fading models, such as Rayleigh (for $m=1$), Rician fading with parameter $K$ (for $m=(K+1)^2/(2K+1)$) and no fading (for $m\rightarrow\infty$), as well as intermediate fading distributions. It is of unit mean and variance and given by
\begin{equation}\label{Eqn:NakagamiFadPDF}
f_H(h)=\frac{m^m}{\Gamma(m)}h^{m-1}\exp(-mh),
\end{equation}
where  $m$ is a positive integer and $\Gamma(m)=\int_{0}^{\infty} z^{m-1}e^{-z}\, \dif z$ is the Gamma function.

Each receiver is able to successfully receive its desired information if its SIR is greater or equal to the target threshold $\beta$. That is, receiver node $Y_j$ is ``connected'' to the typical transmitter if
\begin{equation}\label{Eqn:SINRwThreshold}
\mathrm{SIR} \defn \frac{\,H_j \ell(|Y_j|)}{\,I_j}\geq \beta.
\end{equation}
Note that $\mathrm{SIR}$ depends on $\lambda_{\textsf{t}}$, i.e. $\mathrm{SIR}=\mathrm{SIR}(\lambda_{\textsf{t}})$. All the transmitters are assumed to use the same transmit power, the network is interference-limited, and $I_j$ is the aggregate interference at receive node $Y_j$ and a sum over the marked point processes. Namely,
\begin{equation}\label{Eqn:PoissonShotNoise1}
    I_j=\sum_{X_{i}\in\Phi^{\textsf{t}} \setminus\{X_0\}} \tilde{H}_{ij}\ell(|X_{i}-Y_j|),
\end{equation}
which is a Poisson shot noise process, and $\tilde{H}_{ij}$ is the fading channel gain from transmitter $X_i$ to receiver $Y_j$ in $\mathcal{R}_0$. Since $\Phi^{\textsf{t}}$ is stationary,  according to Slivnyak's theorem the statistics of signal reception seen by receiver $Y_j$ is the same as that seen by any other receivers in the same cluster. Thus $I_j$ can be evaluated at the origin, \ie \eqref{Eqn:PoissonShotNoise1} can be rewritten as $I_0=\sum_{X_i\in\Phi_t\setminus\{X_0\}}\tilde{H}_i \ell(|X_i|)$, where $\tilde{H}_i$ is the fading channel gain between transmitter $X_i$ and the origin.

Suppose the decoding delay is up to the lapse of $\tau$ transmission attempts for a transmitter. The connected receiver process for the $t$-th transmission is denoted by
\begin{equation}
\hat{\Phi}^{\textsf{c}}(t) = \left\{(Y_j,H_j(t))\in\Phi_0^{\textsf{r}}:  H_j(t)\ell(|Y_j|)\geq\beta I_0\right\},
\end{equation}
where $\{H_j(t)\}$ are i.i.d. for all $t\in [1,\cdots,\tau]$. Also, let $\Phi^{\textsf{c}}(\tau)$ be the connected receiver process at the $\tau$th attempt, \ie it is the set of all intended receivers in a cluster connected by their transmitter during the decoding delay, and thus it can be written as $\Phi^{\textsf{c}}(\tau) = \bigcup_{t=1}^{\tau} \hat{\Phi}^{\textsf{c}}(t)$. In other words, the connected receiver process can be described by a \emph{filtration} process\footnote{A filtration process means $\Phi^{\textsf{c}}(1)\subseteq \Phi^{\textsf{c}}(2)\cdots\subseteq \Phi^{\textsf{c}}(\tau)$, and for any set $\mathcal{A}\subseteq \mathcal{R}_0$, $\mathcal{A}(\Phi^{\textsf{c}}(\tau))\rightarrow \mathcal{A}(\Phi_0^{\textsf{r}})$ almost surely as $\tau\rightarrow \infty$ where $\mathcal{A}(\Phi)$ denotes the random number of point process $\Phi$ enclosed in set $\mathcal{A}$.}.

\subsection{Multicast Transmission Outage}\label{Sec:MultiTransOutage}
The transmission capacity of an ad hoc network introduced in \cite{SWXYJGAGDV05} is defined based on point-to-point transmission with an outage probability constraint $\epsilon\in(0,1)$, and is given by
\begin{equation}\label{Eqn:UnicastTC}
\tilde{c}_{\epsilon} = \tilde{b}\,\bar{\lambda}_{\textsf{t}}\,(1-\epsilon),
\end{equation}
where $\tilde{b}$ is the constant transmission rate a communication link can support (for example, about $\log_2(1+\beta)$), and $\bar{\lambda}_{\textsf{t}}$ is the maximum contention intensity subject to an outage probability target $\epsilon$. However, \eqref{Eqn:UnicastTC} cannot be directly applied to multicast because the multicast rate would be affected by $\lambda_{\textsf{r}}$ and $\bar{\lambda}_{\textsf{t}}$, and the outage of a multicast transmission is not point-to-point but \emph{point-to-multipoint}. How to declare an outage event for a transmitter multicasting information in the previous multicast transmission model is a key issue.

Since no desired receiver can be assumed to be dispensable, a reasonable way to define multicast outage is when \emph{any of the intended receivers of a transmitter does not receive a multicasted packet} during a period of time up to the decoding delay. That is, after all the allowed retransmissions have been used, if one of the desired receivers in the cluster has not decoded the packet, we declare an outage for this cluster. Thus, a multicast outage event of each multicast cluster can be described as $\mathcal{E}(\lambda_{\textsf{t}}) =\{\Phi_0^{\textsf{r}}\setminus\Phi^{\textsf{c}}(\tau)\neq \emptyset\}$ because $\mathcal{E}$ depends on $\lambda_{\textsf{t}}$. The probability of $\mathcal{E}(\lambda_{\textsf{t}})$ can be characterized by the intensity of the connected receivers during the lapse of $\tau$ attempts as follows:
\begin{eqnarray}\label{Eqn:DefOutProbMTC1}
\mathbb{P}[\mathcal{E}(\lambda_{\textsf{t}})] &\defn& 1-\mathbb{P}[\{\Phi_0^{\textsf{r}}\setminus\Phi^{\textsf{c}}(\tau)\}=\emptyset]\nonumber\\
&=& 1-\exp\left\{-\int_{\mathcal{R}_0}(\lambda_{\textsf{r}}-\lambda_{\textsf{c}}(Y,\tau))\,\mu(\dif Y)\right\},
\end{eqnarray}
where $\mu(\cdot)$ denotes a  Lebesgue measure and $\lambda_{\textsf{c}}(Y,\tau)$ is the intensity of $\Phi^{\textsf{c}}(\tau)$ at node $Y$. Using \eqref{Eqn:DefOutProbMTC1} to find multicast outage probability can be interpreted as finding the void probability of a ``disconnected'' PPP in a cluster. Since all of the intended receivers are uniformly distributed in $\mathcal{R}_0$, \eqref{Eqn:DefOutProbMTC1} can be rewritten as
\begin{equation}\label{Eqn:DefOutProbMTC2}
\mathbb{P}[\mathcal{E}(\lambda_{\textsf{t}})]=
1-\exp\left\{-\pi s^2\,(\lambda_{\textsf{r}}-\mathbb{E}_{R}[\lambda_{\textsf{c}}(R,\tau)])\right\}\leq \epsilon,
\end{equation}
where $R\in[0,s]$ is a random variable whose PDF is $f_R(r)=\frac{2r}{s^2}$. The outage probability in \eqref{Eqn:DefOutProbMTC2} cannot exceed its designated upper bound $\epsilon$ which is assumed to be a small value throughout this paper.

\begin{remark}
In Section \ref{Sec:RxConnectedPro}, we will show that  $\Phi^{\textsf{c}}(\tau)$ is a nonhomogeneous PPP and the average of its intensity $\mathbb{E}_R[\lambda_{\textsf{c}}(R,\tau)]$ can be found. Thus, the multicast outage probability  in \eqref{Eqn:DefOutProbMTC2} can be calculated. Also,  $\mathbb{E}_R[\lambda_{\textsf{c}}(R,\tau)]$ is a monotonically decreasing function of $\lambda_{\textsf{t}}$ so that the maximum $\bar{\lambda}_{\textsf{t}}$ is attained when  $\mathbb{E}_R[\lambda_{\textsf{c}}(R,\tau)]$ reduces to  its lower bound $\lambda_{\textsf{r}}+\frac{\ln(1-\epsilon)}{\pi s^2}$ that is obtained by solving \eqref{Eqn:DefOutProbMTC2}.
\end{remark}

\begin{remark}
The multicast outage probability in \eqref{Eqn:DefOutProbMTC2} will accurately approximate the unicast outage probability for a small $\epsilon$ if there are no retransmissions ($\tau=1$) and only one receiver is in a cluster. For a unicast scenario in each cluster, we have $\pi s^2\lambda_{\textsf{r}}=1$ and $\mathbb{E}_R[\lambda_{\textsf{c}}(R,1)]=\lambda_{\textsf{r}}\,\mathbb{P}[\mathrm{SIR}(\lambda_{\textsf{t}})\geq \beta]$. Thus, \eqref{Eqn:DefOutProbMTC2} becomes
\begin{equation*}
\mathbb{P}[\mathcal{E}(\lambda_{\textsf{t}})]=1-\exp(-\mathbb{P}[\mathrm{SIR}(\lambda_{\textsf{t}})<\beta])=\mathbb{P}[\mathrm{SIR}(\lambda_{\textsf{t}})<\beta]+O(\epsilon^2)
\end{equation*}
since $\mathbb{P}[\mathrm{SIR}(\lambda_{\textsf{t}})<\beta]\leq \epsilon$ and $e^{-\epsilon}=1-\epsilon+O(\epsilon^2)$. This shows that the point-to-point outage scenario is covered by our model, i.e. multicast outage is a generalization of point-to-point outage.
\end{remark}

\subsection{Definitions of Multicast Transmission Capacity, Largeness and Denseness of Networks}\label{Sec:DefnMTC}
Since the multicast outage probability is upper bounded by a small $\epsilon$, the maximum contention intensity $\bar{\lambda}_{\textsf{t}}$ is a function of $\epsilon$ and its  Taylor expansion for $\epsilon$ gives
\begin{equation}\label{Eqn:ApproxMaxContIntensity}
\bar{\lambda}_{\textsf{t}}(\epsilon)\defn\sup\{\lambda_{\textsf{t}}>0: \mathbb{P}[\mathcal{E}(\lambda_{\textsf{t}})]\leq \epsilon\}  = \bar{\lambda}_{\epsilon}+O(\epsilon^2),
\end{equation}
where $\bar{\lambda}_{\epsilon}$ is the Taylor expansion of $\bar{\lambda}_{\textsf{t}}(\epsilon)$ without the second and higher order terms of $\epsilon$. Since $\epsilon$ is small, $\bar{\lambda}_{\epsilon}\approx\bar{\lambda}_{\textsf{t}}(\epsilon)$ and thus, for simplicity, we will focus the analysis on $\bar{\lambda}_{\epsilon}$ in the following.
\begin{definition}[\textbf{Multicast Transmission Capacity}]\label{Def:MTC}
The multicast transmission capacity with the multicast outage probability defined in \eqref{Eqn:DefOutProbMTC1} for small $\epsilon$ is defined as
\begin{equation}\label{Eqn:MTCwoACK}
 C_{\epsilon} \defn \frac{1}{\tau}\,b\,\bar{\lambda}_{\epsilon}\,(1-\epsilon),
\end{equation}
where $\bar{\lambda}_{\epsilon}$ is the first order approximation of $\bar{\lambda}_{\textsf{t}}(\epsilon)$ as indicated in \eqref{Eqn:ApproxMaxContIntensity}, $b$ is the maximum achievable multicast rate on average for every cluster and it is not a constant in general.
\end{definition}
Multicast transmission capacity $C_{\epsilon}$ gives the number of successful multicast clusters with the maximum achievable multicast rate, that can coexist per unit area subject to decoding delay and multicast outage constraints. In other words, it is the area spectral efficiency of cluster-based multicast transmission. The following definitions of largeness and denseness of a network will be needed to acquire the scaling characterizations of the MTCs in the subsequent analysis. They are defined based on the circumstance that the average number of the intended receivers in a cluster is sufficiently large, \ie $k=\pi s^2 \lambda_{\textsf{r}} \gg 1$.
\begin{definition}[\textbf{Denseness and Largeness of a network with a PCP}]\label{Def:DenseLargeNetwork}
(a) We say a network is ``large'' if the area $\pi s^2$ of a cluster in the network is sufficiently large such that for a fixed intended receiver intensity $\lambda_{\textsf{r}}$ we have $k\gg 1$. (b) If the intended receiver intensity is sufficiently large such that for fixed area $\pi s^2$ we have $k\gg 1$, then such a network is called ``dense''. (c) A ``large dense'' network, it means that clusters in a network have a sufficiently large size as well as intended receiver intensity; namely, $\lambda_{\textsf{r}}\propto \pi s^2$ and thus $k \gg 1$.
\end{definition}
Here we should point out that Definition \ref{Def:DenseLargeNetwork} may not be consistent with some popular node-based unicast definitions in prior literature. For example, a dense network usually means it is dense everywhere (\ie \emph{uniformly} dense); however, our denseness definition could involve the case of \emph{local denseness} if the receiver intensity in a single cluster is sufficiently large whereas the cluster intensity is small.

\subsection{Multicast Transmission Methods -- Single-hop and Multihop}
The network model with a PCP for multicast introduced in the previous subsection implicitly assumes that the parent node of each cluster (see Fig. \ref{Fig:MulticastModel}) is  the sole transmitter. This is the case of single-hop multicast. Hence, for single-hop multicast, multicast outage probability only depends on the channel conditions between the parent node and its intended receiver nodes. When the size of clusters is large, the path loss of transmitted signals and the average number of the intended receivers for each parent node are both increased significantly. The MTC in this case will correspondingly decrease, so single-hop multicast is not an efficient means of disseminating information for a large and/or dense network.

To alleviate this drawback, we propose a multihop multicast approach. The idea is to allow retransmissions in a cluster by randomly selected receivers that have successfully received the packet already. Each of the selected receivers has its own small local multicast region, and the whole cluster is covered by the combination of small multicast regions. Note that only one selected receiver is allowed to transmit for each time slot in order to make all transmitters in each time slot still form a PPP. The detailed algorithm and modeling assumptions will be presented in Section \ref{Sec:MultihopMulticast}.

\section{Multicast Transmission Capacity with Single-Hop Multicast}\label{Sec:MTCwSinglehop}
In this section, we study the MTC when transmitters are multicasting to all of their intended receivers in a single-hop fashion. First we have to find the multicast outage probability defined in \eqref{Eqn:DefOutProbMTC2} and thus we need to study the intensity of the receiver-connected process during the lapse of  $\tau$ attempts. Then the maximum contention intensity which characterizes the single-hop MTC can be found based on the intensity of the connected receivers in a cluster.

\subsection{The Receiver-Connected Process}\label{Sec:RxConnectedPro}
During the allowed transmission $\tau$ attempts, the intended receivers in $\mathcal{R}_0$ connected by the transmitter $X_0$ form a receiver-connected process whose intensity is the necessary information to estimate the multicast outage probability. Since $\Phi^{\textsf{c}}$ is a filtration process and upper bounded by $\Phi^{\textsf{r}}$ (as explained in Section \ref{Sec:MultiTranModelRxConnProc}), the connected receiver intensity in $\mathcal{R}_0$ is an increasing function of $\tau$ as shown in the following lemma.
\begin{lemma}\label{Lem:NonCoopConnRxIntenNakaFading}
Consider the stationary PPP $\Phi^{\textsf{r}}_0$ of intensity $\lambda_{\textsf{r}}$ in the reference cluster $\mathcal{R}_0$. If a transmitter is allowed to transmit a packet up to $\tau$ times, then $\Phi^{\textsf{c}}(\tau)$ is a nonhomogeneous thinning PPP and the intensities of $\Phi^{\textsf{c}}(\tau)$ for different fading models are shown as follows. For Rayleigh fading, we have
\begin{equation}\label{Eqn:NonCoopConnIntenRayFading}
\lambda_{\textsf{c}}(r,\tau) = \lambda_{\textsf{r}}\left(1-\left\{1-\exp\left[-\pi\Delta_1(\beta r^{\alpha},\infty)\lambda_{\textsf{t}}\right]\right\}^{\tau}\right),
\end{equation}
and $\Delta_1(\cdot,\cdot)$ is defined in Proposition \ref{Prop:LapalceShotPPP} in Appendix \ref{App:ProofLaplaceShotPPP}. For Nakagami-$m$ fading with $m>1$, we have
\begin{equation}\label{Eqn:NonCoopConnIntenNakaFading}
\lambda_{\textsf{c}}(r,\tau) = \lambda_{\textsf{r}}\left\{1-\left[1-\Psi^{(m-1)}(m\beta\, r^{\alpha})\right]^{\tau}\right\},\quad r\in[1,s]
\end{equation}
where
\begin{eqnarray}
\Psi^{(m)}(\phi) &\defn& \frac{(-1)^m \phi^{m+1}}{m!}\,\frac{\emph{\dif}^{m}}{\emph{\dif} \label{Eqn:PsiMphi} \phi^{m}}\exp\left\{-\pi\,\lambda_{\textsf{t}}\,\Delta_1(\phi,\infty)-\log\phi\right\}.
\end{eqnarray}
\end{lemma}
\begin{proof}
See Appendix \ref{App:ProofNonCoopConnRxInten}.
\end{proof}

The connection intensity $\lambda_{\textsf{c}}$ for Rayleigh fading in \eqref{Eqn:NonCoopConnIntenRayFading} reveals an interesting implication. The term $1-\exp[-\pi\Delta_1(\beta r^{\alpha},\infty)\lambda_{\textsf{t}}]$ can be interpreted as the probability that there is at least one interferer in the circular area of radius $\sqrt{\Delta_1(\beta r^{\alpha},\infty)}$. This circular area can be called the \emph{dominating interferer area} centered ar $r$ because any single interferer within this area can cause an outage at the receiver located at $r$.The effect of retransmission can be said either to make this probability reduce by $\tau$-fold or to enlarge the circular area. In a dense network with $\lambda_{\textsf{t}}\geq \lambda_{\textsf{r}}$, for example, the term $(1-\exp[-\pi\Delta_1(\beta r^{\alpha},\infty)\lambda_{\textsf{t}}])^{\tau}\approx 1-\tau\exp[-\pi\Delta_1(\beta r^{\alpha},\infty)\lambda_{\textsf{t}}]$ since $\pi\Delta_1(\beta r^{\alpha},\infty)\lambda_{\textsf{t}}\gg 1$. So the radius of the dominating interferer coverage is approximately increased by $\sqrt{\ln(\tau)}$-fold if $\tau>2$ in this case.

In addition, as shown in \eqref{Eqn:NonCoopConnIntenRayFading}-\eqref{Eqn:PsiMphi}, we know that a closed-form expression of the average connection intensity $\mathbb{E}[\lambda_{\textsf{c}}]$ in terms of $\lambda_{\textsf{t}}$ in a general network is difficult to find. However, an upper bound on $\mathbb{E}[\lambda_{\textsf{c}}]$ can be found in the following lemma, and thus in the special case of clusters with many intended receivers, \eqref{Eqn:NonCoopConnIntenRayFading} and \eqref{Eqn:NonCoopConnIntenNakaFading} can be simplified to allow a nearly closed-form solution of $\lambda_{\textsf{t}}$ (see Section \ref{Sec:SingleHopMTC}).
\begin{lemma}\label{Lem:BoundsAveConnInten}
The connection intensity $\lambda_{\textsf{c}}$ of a non-homogeneous PPP $\Phi^{\textsf{c}}(\tau)$ in  a cluster for Rayleigh fading is shown in \eqref{Eqn:NonCoopConnIntenRayFading} and its expression for Nakagami-m fading with $m>1$ is given by \eqref{Eqn:NonCoopConnIntenNakaFading}. The upper bounds on $\lambda_{\textsf{c}}$ for different fading models are given in the following:
\begin{eqnarray}
\mathbb{E}_R[\lambda_{\textsf{c}}(R,\tau)] &\leq& \lambda_{\textsf{r}}\left[1-(1-\mathbb{E}_R[\exp(-\pi\lambda_{\textsf{t}}\Delta_1(\beta R^{\alpha},\infty))])^{\tau}\right] ,\,\,\text{for Rayleigh fading}\label{Eqn:AvgConnIntenRayleighFading}\\
\mathbb{E}_R[\lambda_{\textsf{c}}(R,\tau)] &\leq& \lambda_{\textsf{r}}\left[1-(1-\mathbb{E}_R[\Psi^{(m-1)}(m\beta R^{\alpha})])^{\tau}\right] ,\,\,\text{for Nakagami-m fading, $m>1$}.\label{Eqn:AvgConnIntenNakagami}
\end{eqnarray}
\end{lemma}
\begin{proof}
According to the H\"{o}lder inequality, for two real-valued random variables $D_1$ and $D_2$, $(\mathbb{E}[D_1D_2])^p \leq \mathbb{E}[D^p_1]\mathbb{E}[D^q_2]^{p/q}$, where $p,q>0$ and $1/p+1/q=1$. Consider the number of transmission attempts $\tau>1$ and $D_2=1$,  and let $p=\tau$ and $q=\frac{\tau}{\tau-1}$. Then it follows that $(\mathbb{E}[D_1])^{\tau}\leq \mathbb{E}[D_1^{\tau}]$. For Rayleigh fading, taking average on the both sides of \eqref{Eqn:NonCoopConnIntenRayFading} and using  this property $(\mathbb{E}[D_1])^{\tau}\leq \mathbb{E}[D_1^{\tau}]$ by letting $D_1 = 1-\exp(-\pi\lambda_{\textsf{t}}\Delta_1(\beta R^{\alpha},\infty))$, \eqref{Eqn:AvgConnIntenRayleighFading} follows. Similarly, \eqref{Eqn:AvgConnIntenNakagami} can be shown in the same way.
\end{proof}

\subsection{Single-hop Multicast Transmission Capacity}\label{Sec:SingleHopMTC}
Now we characterize the MTC in a network when a transmitter directly multicasts its intended receivers in a cluster.
\begin{theorem}\label{Thm:MaxContenIntenNonCoop}
Suppose the multicast outage probability given in \eqref{Eqn:DefOutProbMTC2} is upper bounded by small $\epsilon$ and the maximum decoding delay is $\tau$ transmission attempts. If the average number of the intended receivers in a cluster with radius $s$ is $k$ and $k\geq \frac{1}{\epsilon^{\tau-1}}$, then the maximum contention intensity is
\begin{equation} \label{Eqn:MaxContenIntenNoCoop}
\bar{\lambda}_{\epsilon}=\frac{\eta\,\rho\,\tau^2}{s^2\beta^{\frac{2}{\alpha}}\,\sqrt[\tau]{k}} =\Theta\left(\frac{\rho\,\tau^2}{s^2\beta^{\frac{2}{\alpha}}\,\sqrt[\tau]{k}}\right),
\end{equation}
where $\rho=\frac{1}{\tau^2}\sqrt[\tau]{\epsilon(\tau+1)}$, $\beta$ is the SIR threshold for successfully decoding and  $\eta$ is a constant (depending on $m$, $\beta$ and $\alpha$).
\end{theorem}
\begin{proof}
See Appendix \ref{App:ProofMaxContenIntenNonCoop}.
\end{proof}

\begin{remark}
The scaling function $\Theta(\cdot)$ of $\bar{\lambda}_{\epsilon}$ in \eqref{Eqn:MaxContenIntenNoCoop} only contains the ``controllable'' network parameters such as $s$, $\epsilon$, $\tau$ and $\beta$, which means their values are adjustable if needed. Constant $\eta$ contains the parameter $m$ of Nakagami fading, which is a channel characteristic and usually uncontrollable and thus $\eta$ is not left in $\Theta(\cdot)$.
\end{remark}

If a unicast planar network without retransmission is considered (i.e.  $k=\tau=1$), $\bar{\lambda}_{\epsilon}$ in \eqref{Eqn:MaxContenIntenNoCoop} will reduce to the previous results discovered, \ie $\bar{\lambda}_{\epsilon}=\Theta\left(\frac{\epsilon}{s^2\,\beta^{2/\alpha}}\right)$. In \cite{SWXYJGAGDV05}, for example, the maximum contention intensities of FH-CDMA and DS-CDMA are $\Theta\left(\frac{\epsilon M}{s^2\beta^{2/\alpha}}\right)$ and $\Theta\left(\frac{\epsilon M^{2/\alpha}}{s^2\beta^{2/\alpha}}\right)$ respectively, where $M$ is the channel number of FH-CDMA and the spreading factor of DS-CDMA. It is easy to check that these two results coincide with ours here by considering $\frac{\bar{\lambda}_{\epsilon}}{M}$ for FH-CDMA and $\frac{\beta}{M}$ for DS-CDMA. In addition, the longest transmission distance in a cluster is $s$ and we know $\bar{\lambda}_{\epsilon}=\Theta(s^{-2})$ and so is the network capacity, which also coincides with the results in \cite{SWXYJGAGDV05}\cite{FBBBPM06}.

The result in \eqref{Eqn:MaxContenIntenNoCoop} only indicates how much the maximum intensity of transmitters can be supported in a network under the decoding delay and multicast outage constraints. Having the maximum contention intensity only is unable to tell us how much its corresponding MTC should be since the multicast rate $b$ is also affected by the maximum contention intensity. Considering a capacity-approaching code is used, the maximum achievable multicast rate $b$ that is acceptable for all intended receivers is the following ergodic channel capacity evaluated at the boundary of a cluster:
\begin{equation}\label{Eqn:MaxBcRateWoTimeDiv}
b = \mathbb{E}\left[\log \left(1+\frac{H_{\max}\,s^{-\alpha}}{I_0}\right)\right].
\end{equation}
where $H_{\max}= \max_{t\in[1,\cdots,\tau]}H(t)$ and $H(t)$ is the fading channel gain for the $t$-th transmission between typical transmitter $X_0$ and a receiver located on the cluster boundary. Although there may be no receivers on the boundary of a cluster, multicast rate $b$ should be considered from a worse case point of view because it needs to be acceptable for all intended receivers in any locations within a cluster. The bounds on $b$ are given in the following lemma.
\begin{lemma}\label{Lem:NonCoopBoundsMCrate}
There exists a $\delta\in(0,1)$ such that the bounds on the multicast rate $b$ in \eqref{Eqn:MaxBcRateWoTimeDiv} can be given by
\begin{equation}\label{Eqn:BoundsNonCoopBCRate}
\delta\,\log\left(1+\frac{1}{\pi s^2\lambda_{\textsf{t}}}\right) \leq b\leq\log\left(1+\frac{1}{\pi s^2\lambda_{\textsf{t}}}\right)+O(1).
\end{equation}
\end{lemma}
\begin{proof}
See Appendix \ref{App:ProofBoundsMCrate}.
\end{proof}

\begin{remark}\label{Rem:NonCoopBoundsMCrate}
The bounds on the multicast rate in \eqref{Eqn:BoundsNonCoopBCRate} are not affected by channel fading because the fading effect has been averaged out. Lemma \ref{Lem:NonCoopBoundsMCrate} suggests that $b$ is significantly reduced by the aggregate interference from the transmitters in the cluster and the transmitters out of the cluster only reduce it by at most a constant.
\end{remark}

\textbf{Scaling Law of Single-hop MTC}. According to Theorem \ref{Thm:MaxContenIntenNonCoop} and Lemma \ref{Lem:NonCoopBoundsMCrate}, we found that the multicast rate $b$ is $\Theta(\log(k)/\tau)$ for any network conditions if $\bar{\lambda}_{\epsilon}$ is achieved and $k$ is sufficiently large. The MTCs in a network without receiver cooperation can be concluded as follows. (\textbf{i}) For a dense network, $\bar{\lambda}_{\epsilon}=\Theta\left(\frac{\rho\tau^2}{\sqrt[\tau]{k}}\right)$ since $\pi s^2$ is fixed and $k\gg 1$. By the MTC definition, we know $C_{\epsilon}=\Theta\left(\frac{\rho\log(k)}{\sqrt[\tau]{k}}\right)$. (\textbf{ii}) If the network is large, then $\bar{\lambda}_{\epsilon}=\Theta\left(\frac{\rho\tau^2}{k^{1+1/\tau}}\right)$ and thus $C_{\epsilon}=\Theta\left(\frac{\rho\log(k)}{k^{(1+1/\tau)}}\right)$. (\textbf{iii}) For a large dense network, $\lambda_{\textsf{r}}=\Theta(s^2)$ and $k\gg 1$. So $\bar{\lambda}_{\epsilon}$ is $\Theta\left(\rho\,\tau^2\,k^{-\left(\frac{\tau+2}{2\tau}\right)}\right)$ and thus $C_{\epsilon}=\Theta\left(\rho k^{-\left(\frac{\tau+2}{2\tau}\right)}\log(k)\right)$. In summary, the MTC here can be expressed in a general form as follows:
\begin{equation}\label{Eqn:GeneralMTC}
 C_{\epsilon} = \Theta\left(\rho\,k^{x}\log(k)\right),
\end{equation}
where $x$ has been given in Table \ref{Tab:MainResultsMTC}. Note again that the scaling function $\Theta(\cdot)$ is applied only to the main controllable variables for a specific network. For example, the main controllable parameters for a dense network are $\lambda_{\textsf{r}}$ (or $k$) and $\tau$ whereas $\pi s^2$ is a constant.

A simulation example of the MTCs for a large network with Rayleigh fading is presented in Fig. \ref{Fig:ShMTCwoRxCoop1}. One can see that the MTCs decrease when $s$ increases (i.e. $k$ increases), and slightly increasing the radius of a cluster can significantly reduces the MTC. The decoding delay constraint $\tau$ has a significant effect on the MTCs as well. Since retransmissions decrease outage probability as well as increase interference, it can be observed from \eqref{Eqn:MTCwoACK} that there exists an optimal tradeoff between $\tau$ and $\bar{\lambda}_{\epsilon}$. This can be observed in Fig. \ref{Fig:ShMTCwoRxCoop1}, for example, the MTC of $\tau=3$ is larger than that of $\tau=10$, which indicates a few retransmissions indeed increase the MTC and too many retransmissions are, on the contrary, detrimental to it. Fig. \ref{Fig:ShMTCwoRxCoop2} presents the MTC results for a large dense network with Rayleigh fading. If we compare Fig. \ref{Fig:ShMTCwoRxCoop2} with Fig. \ref{Fig:ShMTCwoRxCoop1}, we can see that denseness does not have a serious impact on MTC as largeness. This observation coincides with the scaling law stated in above. Thus, path loss is the main key issue of limiting the MTC, which enlightens us the idea of using multihop multicast to deliver packets with a less path loss (see Section \ref{Sec:MTCwMultihop}).

\section{Multicast Transmission Capacity with Multihop Multicast}\label{Sec:MTCwMultihop}
The MTCs for single-hop multicast investigated in Section \ref{Sec:MTCwSinglehop} scale like $\Theta(s^{-2})$ if other network parameters are fixed. Therefore, the single-hop MTC increases when its multicast cluster is shrunk. However, shrinking the cluster is not welcome if the packets must be transmitted over the same coverage. So here we would like to know if there is another method to increase the MTC without shrinking the cluster.
From previous results, for sufficiently large $k$ the scaling of the single-hop MTC can be written as
\begin{equation}\label{Eqn:SingleHopMTC}
  C_{\epsilon}=\Theta\left(\frac{\rho\,k^{x}\log(k)}{s^2\,\beta^{\frac{2}{\alpha}}}\right),
\end{equation}
where $\rho=\frac{1}{\tau^2}\sqrt[\tau]{\epsilon(\tau+1)}$ and $x=-\frac{1}{\tau}$.
So \eqref{Eqn:SingleHopMTC} suggests three approaches to increasing the MTC: interference-avoidance, interference-suppression and area-shrinking methods. The capacity gain due to interference avoidance  can be acquired by removing co-channel interferers such that $\mathrm{SIR}(\lambda_{\textsf{t}})$ is improved by reducing $\lambda_{\textsf{t}}$. This is the context when each transmitter independently selects its own transmission channel from several available channels (see the case of FH-CDMA in \cite{SWXYJGAGDV05}). Interference avoidance does not affect any network parameters except the multicast rate $b$. The interference-suppression capacity gain can be obtained by signal processing techniques to increase the SIR so that multicast rate $b$ is increased (see the case of DS-CDMA in \cite{SWXYJGAGDV05}). In addition, suppressing interference is  equivalent to relaxing the SIR  threshold $\beta$ and thus pre-constant $\eta$ depending on $\beta$  is increased. Of course, the area-shrinking capacity gain is attained by multicasting in a smaller region instead of the whole cluster. Shrinking a cluster only leads to a decrease in the cluster radius $s$ and thus the average number of the intended receivers $k$ becomes smaller.

Let $\{g_a,g_s,g_v\}> 1$ be the interference-avoidance, interference-suppression and area-shrinking gain parameters, respectively. Then respectively replacing $s^2$, $\beta$ and $\bar{\lambda}_{\epsilon}$ by $s^2/g_v$, $\beta/g_s$ and $\bar{\lambda}_{\epsilon}/g_a$ in \eqref{Eqn:SingleHopMTC}, we obtain
\begin{equation}\label{Eqn:GenScalingMTCwGain}
C_{\epsilon}=\Theta\left(\frac{\rho\,g_a\,g_s^{2/\alpha}\,g_v^{1-x}\,k^{x}\,\log(k)}{s^2\beta^{\frac{2}{\alpha}}}\right).
\end{equation}
The three gain parameters in \eqref{Eqn:GenScalingMTCwGain} indicate which method is able to contribute more to the MTC in each context. The interference-avoidance method is superior to interference-suppression when $g_a$ and $g_s$ are equal because $g_a/g_s^{2/\alpha}>1$ due to $\alpha>2$. This point has been shown for FH-CDMA and DS-CDMA in \cite{SWXYJGAGDV05}. If the cluster is shrunk to a smaller region of area $\pi s^2/g_v$ and all the three gain parameters are equal, then shrinking is the best way to improve the MTC. According to \eqref{Eqn:GenScalingMTCwGain}, we can conjecture that the MTC could be increased if the cluster is tessellated into several smaller regions and a packet is allowed to be multicasted in each of them up to some times under the condition that the decoding delay and multicast outage constraints both have to be satisfied. This conjecture will be verified later in the following subsection.

\subsection{Multicast over Multihop}\label{Sec:MultihopMulticast}
Suppose the cluster $\mathcal{R}_0$ is tessellated into a certain number of smaller multicast regions of equal area. A packet is multicasted the same number of times (called a multicast time slot) in each tessellated region. So the packet is delivered slot by slot from the central typical transmitter to those regions in a certain order. Note that the number of the tesselated regions cannot exceed the decoding delay constraint $\tau$, \ie the packet is delivered to all its intended receivers at most $\tau-1$ hops. As shown in Fig. \ref{Fig:MultihopTransModel}, for example, each cluster consists of 6 smaller tessellated regions and the decoding delay constraint for each region is 2 if $\tau=12$.

Delivering a packet by the multihop multicast method proceeds as follows. In the first time slot, the typical transmitter aims at multicasting its own region. In the second time slot, the typical transmitter randomly selects a receiver in the neighboring regions that successfully received the packet in the previous time slot, then the selected receiver becomes the transmitter for the next time slot. The multihop multicast method proceeds slot by slot in this way until all the tessellated regions are visited by the packet. Note that in each time slot the transmitter is asked to ensure all the receivers in its region and at least one receiver in the neighboring regions should receive the packet; otherwise, there is an outage. In addition, although there might be more than one receivers in the neighboring regions which successfully receive the packet; however, only one of them is selected to multicast in the next time slot in order to satisfy the assumption of the PPP of the transmitters (see Proposition \ref{Prop:DualityPcpPpp} in Appendix \ref{App:DualityPcpPpp} for the dual property between PCP and PPP). The scaling realizations of the maximum contention intensity for the above multihop multicast method are shown in the following theorem.
\begin{theorem}\label{Thm:MultihopMulticastMaxInten}
Suppose cluster $\mathcal{R}_0$ is tessellated into $v$ smaller multicast regions of equal area and each packet is allowed to be multicasted in each region at most $\tau/v$ times. If the average number of the intended receivers in a cluster $k\geq\frac{v}{\epsilon^{\tau/v-1}}$, then the following scaling characterization of the maximum contention intensity $\bar{\lambda}_{\epsilon}$ is achieved with high probability:
\begin{equation}\label{Eqn:MultihopMaxContenInten}
\bar{\lambda}_{\epsilon}=\frac{\eta\, k^{x}v^y\,\tau^2\,\rho}{s^2\beta^{\frac{2}{\alpha}}}=\Theta\left(\frac{k^{x}v^y\,\tau^2\,\rho}{s^2\beta^{\frac{2}{\alpha}}}\right),
\end{equation}
where $x=-\frac{v}{\tau}$, $y=\frac{v}{\tau}+1$, $\rho=\frac{1}{\tau^2}(\epsilon(\tau/v+1))^{\frac{v}{\tau}}$ and $\eta$ is a constant depending on $m$, $\beta$ and $\alpha$.
\end{theorem}
\begin{proof}
See Appendix \ref{App:ProofSectorMaxContenInten}.
\end{proof}

By comparing \eqref{Eqn:MaxContenIntenNoCoop} and \eqref{Eqn:MultihopMaxContenInten}, it is hard to see if the multihop multicast method achieves a larger $\bar{\lambda}_{\epsilon}$. Nevertheless, in the following subsection we will show that $\bar{\lambda}_{\epsilon}$ and its corresponding MTC indeed achieve a larger value by multihop multicast. The simulation result of the MTC with multihop multicast in a large network with Rayleigh fading and $\lambda_{\textsf{r}}=0.2$ is presented in Fig. \ref{Fig:MhMTCwRxCoop1}, and it is easy to observe that tessellating the cluster into a certain number of regions improves the MTC whereas too many tessellations degrades it. Fig. \ref{Fig:MhMTCwRxCoop2} shows the simulation of the MTC achieved by multihop multicast in a large dense network with Rayleigh fading. If we compare Fig. \ref{Fig:MhMTCwRxCoop2} with Fig. \ref{Fig:MhMTCwRxCoop1}, we can find that the curves in these two figures are not very much different. This point is quite different from the case of single-hop multicast (see Figs. \ref{Fig:ShMTCwoRxCoop1} and \ref{Fig:ShMTCwoRxCoop2}), and therefore, it reveals that multihop multicast is able to efficiently alleviate the impact on MTC due to denseness.  How to tessellate the cluster to achieve a larger MTC will be discussed in the following subsection. In addition, there exists another time-division multicast method to multicast in Fig. \ref{Fig:MultihopTransModel}. Namely, the typical transmitter multicasts a smaller region in each time slot. We can show that this time-division multicast method attains a less $\bar{\lambda}_{\epsilon}$ than the multihop multicast method. Since it just uses the same transmitter to multicast, the average distance from the transmitter to the intended receivers is longer than that of the multihop multicast method. So time-division multicast has a larger path loss so that the receivers have a lower SIR and thus less cluster transmissions are allowed.

\subsection{Capacity Gain Achieved by Multihop Multicast}\label{Sec:CapGainMultihopMulticast}
According to Theorem \ref{Thm:MultihopMulticastMaxInten}, the MTC with multihop multicast obtained from \eqref{Eqn:MultihopMaxContenInten} and $b=\Theta(\log(k)/\tau)$ can be concluded as follows
\begin{equation}\label{Eqn:MultihopMTC}
C_{\epsilon} = \Theta\left(\frac{\rho\,k^{x}\,v^{y}\,\log(k)}{s^2 \beta^{\frac{2}{\alpha}}}\right).
\end{equation}
Since just comparing \eqref{Eqn:MultihopMTC} with its single-hop counterpart \eqref{Eqn:SingleHopMTC} is hard to see if multihop multicast achieves a larger MTC or not, defining a capacity gain $g_c$ as follows can help us understand when multihop multicast is better than single-hop multicast.
\begin{eqnarray}\label{Eqn:CapaGainMTC}
 g_c(v) \defn 10\log_{10}\left[\frac{C_{\epsilon}\text{ in }\eqref{Eqn:MultihopMTC}}{C_{\epsilon}\text{ in }\eqref{Eqn:SingleHopMTC}}\right],\,\,(\text{dB}).
\end{eqnarray}
Since the capacity gain is dependent on $v$, we can formulate an optimization problem with constraints on $v$ as follows.
\begin{eqnarray}\label{Eqn:OptProbCapaGain}
\min_{v}\, -g_c(v),\quad  \text{subject to }\, 1-v\leq 0\,\text{ and }\,v-\tau \leq 0.
\end{eqnarray}
Hence, if there exists a minimizer $v$ such that the minimum of $-g_c$ is negative, then multihop multicast achieves a higher MTC than single-hop multicast. The problem in \eqref{Eqn:OptProbCapaGain} is a convex optimization problem as shown in the following theorem.
\begin{theorem}\label{Thm:OptProbCapacityGain}
Suppose a packet is delivered by the multihop multicast method and each cluster is tessellated into $v$ smaller multicast regions of equal area. \eqref{Eqn:OptProbCapaGain} is a convex optimization problem and thus there exists a unique feasible solution of $v$ to it.
\end{theorem}
\begin{proof}
We have to verify that the optimization problem in \eqref{Eqn:OptProbCapaGain} is convex and it has only one optimal solution of $v$. In other words, we have to show that $-g_c$ is strictly convex. Now considering the noncooperative receiver case, the capacity gain in \eqref{Eqn:CapaGainMTC} can be explicitly expressed as follows:
\begin{eqnarray*}
\frac{g_c}{10} &=& \left(\frac{1-v}{\tau}\right)\log_{10}\left(\frac{k}{\epsilon}\right)+\frac{v}{\tau}\log_{10}\left(1+\frac{v}{\tau}\right)
+\left(\frac{v}{\tau}+2\right)\log_{10}(v)-\frac{\log_{10}(1+\tau)}{\tau}\\
&=& z\,\log_{10}\left(\frac{\tau\,\epsilon}{k}\right)+z\,\log_{10}(1+z)+(2+z)\log_{10}(z)+
2\,\log_{10}(\tau)+\frac{1}{\tau}\log_{10}\left(\frac{k}{\epsilon(1+\tau)}\right),
\end{eqnarray*}
where $z=\frac{v}{\tau}$. Taking the derivative of the above equation with respect to $z$ twice and letting $\frac{\dif^2 g_c}{\dif z^2}<0$, it yields the inequality $(1+\tau/v)\sqrt{(2\tau-v)/(2\tau+v)}>1$,
which is always true for all $\tau\in\mathbb{N}_+$. Thus, $-g_c$ is strictly convex.
\end{proof}

Theorem \ref{Thm:OptProbCapacityGain} indicates that the proposed multihop multicast method can achieve a larger MTC than single-hop multicast since $g_c$ is strictly concave. The reason that multicasting in the smaller regions hop by hop can increase the MTC is because the average number of intended receivers in a tessellated region and path loss are largely reduced, and thus merely few retransmissions are able to make the multicast outage probability lower its designated upper bound. So the increase in $\bar{\lambda}_{\epsilon}$ is over the loss in spectral efficiency. For time-division multicast, no capacity gain is achieved by it since its capacity gain function is also negative convex.

\section{Conclusions}\label{Sec:Conclusion}
Multicast transmission in a wireless ad hoc network is modeled by a Poisson cluster process, where transmitters in the network follow a stationary PPP and each of them is associated with an area-fixed cluster in which the intended receivers are a another stationary PPP. The MTC is defined under the constraints on the multicast outage probability and the decoding delay. Three network conditions, dense, large, and large dense, are specified in order to attain scaling characterizations on MTC. The scaling behaviors of the single-hop and multihop MTCs under the three network conditions are presented by a general expression. They are affected significantly by the decoding delay but not by the fading channel models. In addition, for multihop multicast our main result shows that the MTC is superior to its single-hop counter part if all clusters are tessellated appropriately.


%
\appendices
\section{Useful Propositions}
\subsection{Moment Generating Functional of Stationary Independent PPPs}\label{App:ProofLaplaceShotPPP}
\begin{proposition}\label{Prop:LapalceShotPPP}
Let $\mathcal{B}(0,r)$ be a circular disc centered at the origin with radius $r$ and $\Phi_i=\{(X_{ij},H_{ij}):X_{ij}\in\mathcal{B}(0,r)\cap\mathbb{R}^2, r\geq 1,j\in\mathbb{N}\}$ be a stationary marked PPP of intensity $\lambda_i$ for all $i\in[1,2,\cdots,L]$ and $\{\sqrt{H}_{ij}\}$ are i.i.d. Nakagami-$m$ random variables with unit mean and variance. Suppose $\Phi_i$ has a Poisson shot generating function $I_i:\mathbb{R}_+^2\times \mathbb{R}_+\rightarrow \mathbb{R}_+$ which is defined as $I_i\defn \sum_{X_i\in\Phi_i} H_{ij}\ell(|X_{ij}|)$ where $\alpha>2$. If $\{\Phi_i\}$ are independent, then the sum of the Poisson shot generating functions, i.e. $I=\sum_{i=1}^L I_i$, has the following moment generating functional:
\begin{eqnarray}
\mathcal{L}_I(\phi_1)&=&\mathbb{E}\left[e^{-\phi_1 I}\right]= \exp\left(-\pi\,\Delta_1(\phi_1,r)\,\sum^L_{i=1} \lambda_i\right),\,\, \forall \phi_1\in\mathbb{R}_{++}\label{LaplaceShotPPP1}\\
\mathcal{M}_I(\phi_2)&=&\mathbb{E}\left[e^{\phi_2 I}\right]= \exp\left(\pi\,\Delta_2(\phi_2,r)\,\sum^L_{i=1} \lambda_i\right),\,\, \forall \phi_2\in\left(0,m\,r^{\alpha}\right)\label{LaplaceShotPPP2}.
\end{eqnarray}
where
\begin{eqnarray}
\Delta_1\left(\phi_1,r\right)&\defn& \frac{2}{\alpha} \left(\frac{\phi_1}{m}\right)^{ \frac{2}{\alpha}}\sum_{j=0}^{m-1} {m\choose j}\int_{m/\phi_1}^{m r^{\alpha}/\phi_1} \frac{t^{j-1+ \frac{2}{\alpha}}}{(1+t)^m}\dif t, \label{Eqn:Delta1}\\
\Delta_2\left(\phi_2,r\right) &\defn&  \frac{2}{\alpha} \left(\frac{\phi_2}{m}\right)^{ \frac{2}{\alpha}}\sum_{j=0}^{m-1} (-1)^j{m\choose j} \int_{m/\phi_2}^{m r^{\alpha}/\phi_2}\frac{ t^{j+ \frac{2}{\alpha}}}{(1-t)^m}\,\dif t.\label{Eqn:Delta2}
\end{eqnarray}
\end{proposition}
\begin{proof}
Since the marks $\{H_{ij}\}$ for all points $\{X_{ij}\}$ in their corresponding PPP are i.i.d. Nakagami-$m$ random variables with unit mean and variance, their probability density function is
given in \eqref{Eqn:NakagamiFadPDF} and rewritten in below for convenience:
\begin{equation*}
f_H(h)=\frac{m^m}{\Gamma(m)}h^{m-1}\exp(-mh).
\end{equation*}
Thus the Laplace transform of $f_H(h)$ is
\begin{eqnarray*}
  \mathcal{L}_H(w) &=& \int_0^{\infty} e^{-wh} f_H(h)\, \dif h\\
  &=&  \frac{m^m}{\Gamma(m)}\int_0^{\infty} e^{-(w+m)h}h^{m-1}\, \dif h = \frac{1}{(1+w/m)^{m}}.
\end{eqnarray*}
Moreover, the Laplace transform for a Poisson shot process $\Phi_i$ with i.i.d. marks $\{H_{ij}\}$ is given by \cite{JFCK93}
\begin{equation}\label{Eqn:LapTrsfmPosShitNoise}
\mathcal{L}_{I_i}(\phi_1)=\exp\left\{-\lambda_i \int_{\mathcal{B}(0,r)}\left(1-\mathbb{E}_H\left[e^{-\phi_1 H \ell(|X|)}\right]\right)\dif X\right\},
\end{equation}
and we also know
\begin{equation}\label{Eqn:ExpFadingChannel}
\mathbb{E}_H\left[e^{-\phi_1 H \ell(|X|)}\right]=\int_{0}^{\infty} e^{-\phi_1 \ell(|X|) h} f_H(h)\,\dif h=\mathcal{L}_H\left(\phi_1 \ell(|X|)\right).
\end{equation}
Substituting \eqref{Eqn:ExpFadingChannel} into \eqref{Eqn:LapTrsfmPosShitNoise} and it follows that
\begin{eqnarray}
\mathcal{L}_{I_i}(\phi_1) &=& \exp\left(-\lambda_i\, \int_{\mathcal{B}(0,r)} \left[1-\left(1+\frac{\phi_1}{m}\ell(|X|)\right)^{-m}\right] \,\dif X\right)\nonumber \\
&=& \exp\left(-2\pi \lambda_i\,\left\{\int_{1}^{r} \left[1-\left(\frac{mx^{\alpha}}{m\,x^{\alpha}+\phi_1}\right)^m\right]x \, \dif x\right\} \right)\nonumber\\
&=& \exp\left(-\pi\lambda_i\, \left[ \frac{2}{\alpha} \left(\frac{\phi_1}{m}\right)^{\frac{2}{\alpha}}\sum_{j=0}^{m-1} {m\choose j}
 \int_{m/\phi_1}^{mr^{\alpha}/\phi_1} \frac{t^{j-1+\frac{2}{\alpha}}}{(1+t)^m}\, \dif t\right] \right)\nonumber\\
 &=&\exp(-\pi\lambda_i\,\Delta_1(\phi_1,r))\label{Eqn:Laplacefun1}.
\end{eqnarray}
Since all the PPPs are independent, it yields the following desired result:
$$\mathcal{L}_{I}(\phi_1)=\prod_{i=1}^L \mathbb{E}\left[e^{-\phi_1 I_i}\right]=\exp\left(-\pi\Delta_1(\phi_1,r)\sum_{i=1}^L\lambda_i\right).$$
Now consider the case of $\mathcal{M}_I(\phi_2)$. Similar to $\mathcal{L}_{I_i}(\phi_1)$, $\mathcal{M}_{I_i}(\phi_2)$ can be written as
\begin{equation}
\mathcal{M}_{I_i}(\phi_2)=\exp\left\{\lambda_i\int_{\mathcal{B}(0,r)}\left(1-\mathbb{E}_H\left[e^{\phi_2 H \ell(|Y|)}\right]\right)\dif Y\right\}.
\end{equation}
We also know $\mathcal{M}_H(w)=\int_0^{\infty} e^{wh}f_H(h)\,\dif h=\frac{1}{(1-w/m)^m}$ if $w\in(0,m)$ and $\mathbb{E}_H[\exp(\phi_2 H \ell(|Y|))]=\mathcal{M}_H(\phi_2\ell(|Y|))$ because $\phi_2\in(0,m\,r^{\alpha})$. Then following the same steps in \eqref{Eqn:Laplacefun1}, we can show that
\begin{eqnarray}
\mathcal{M}_{I_i}(\phi_2)&=&\exp\left(\pi\lambda_i\left[\frac{2}{\alpha} \left(\frac{\phi_2}{m}\right)^{\frac{2}{\alpha}}\sum_{j=0}^m {m\choose j}(-1)^j \int_{m/\Phi_2}^{m r^{\alpha}/\Phi_2}\frac{t^{j+\frac{2}{\alpha}}}{(1-t)^m}\dif t\right]\right)\nonumber\\
&=& \exp\left(\pi\lambda_i\,\Delta_2(\phi_2,r)\right).
\end{eqnarray}
Therefore, $\mathcal{M}_I(\phi_2)=\exp\left(\pi\,\Delta_2(\phi_2,r)\sum_{i=1}^L \lambda_i\right)$.
\end{proof}

\subsection{ The Duality between PCP and PPP}\label{App:DualityPcpPpp}
\begin{proposition}\label{Prop:DualityPcpPpp}
A stationary Poisson cluster process (PCP) can be constructed by a given stationary Poisson point process (PPP). Similarly, for a given Poisson cluster process (PCP) a Poisson point process can be constructed from it as well.
\end{proposition}
\begin{proof}
The first statement is based on the definition of a PCP\cite{DSWKJM96}. Consider a stationary PPP $\Phi^{\textsf{d}}$ of intensity $\lambda_{\textsf{d}}$ is given. Then we replace each point $X_j\in\Phi^{\textsf{d}}$ with a random finite set of points $Z_{X_j}$ which is called the cluster associated with point $X_j$. Then the superposition of all clusters yields the stationary PCP $\mathcal{Z}^{\textsf{d}}=\bigcup_{X_j\in\Phi^{\textsf{d}}} Z_{X_j}$. Suppose now a stationary PCP $\mathcal{Z}^{\textsf{d}}$ is given. In the following we are going to show that a PPP can be constructed by randomly selecting a point in each cluster. Let $\mathcal{A}$ be a bounded Borel set, $\Phi^{\textsf{d}}$ be the PPP formed by all parent points in $\mathcal{Z}^{\textsf{d}}$, and $\tilde{\Phi}^{\textsf{d}}$ consists of the points randomly selected from all clusters of $\mathcal{Z}^{\textsf{d}}$. The capacity functional of $\tilde{\Phi}^{\textsf{d}}$ is defined as follows \cite{DSWKJM96}:
\begin{equation}
\tilde{T}_{\textsf{d}}(\mathcal{A}) \defn \mathbb{P}\left[\mathcal{A}(\tilde{\Phi}^{\textsf{d}})>0\right],
\end{equation}
where $\mathcal{A}(\tilde{\Phi}^{\textsf{d}})$ denotes the numbers of points of $\tilde{\Phi}^{\textsf{d}}$ in a bounded set $\mathcal{A}\subset\mathbb{R}^2$. Since the capacity functional of a point process completely characterizes its distribution, $\tilde{\Phi}^{\textsf{d}}$ is a stationary PPP if we can show
\begin{equation*}
\tilde{T}_{\textsf{d}}(\mathcal{A})=\mathbb{P}[\mathcal{A}(\Phi^{\textsf{d}})>0]=1-\exp(-\lambda_{\textsf{d}}\,\mu(\mathcal{A})),
\end{equation*}
where $\mu(\mathcal{A})$ denotes the Lebesgue measure of set $\mathcal{A}$.

Consider a large bounded Borel set $\mathcal{C}$ such that $\mathcal{A}\subset \mathcal{C}$ and $\mathcal{C}$ encloses all clusters of $\mathcal{Z}^{\textsf{d}}$. So $\mathcal{C}(\Phi^{\textsf{d}})$ is a Poisson random variable of parameter $\lambda_{\textsf{d}}\mu(\mathcal{C})$. Since $\tilde{\Phi}^{\textsf{d}}$ is formed by the points that randomly selected by all points in $\Phi^{\textsf{d}}$, the probability that all points in $\tilde{\Phi}^{\textsf{d}}$ are not in $\mathcal{A}$ is
\begin{eqnarray*}
\mathbb{P}\left[\mathcal{A}(\tilde{\Phi}^{\textsf{d}})=0\right]&=&\mathbb{P}\left[\mathcal{A}(\tilde{\Phi}^{\textsf{d}})=0\big|\mathcal{C}(\tilde{\Phi}^{\textsf{d}})=0\right]
\mathbb{P}\left[\mathcal{C}(\tilde{\Phi}^{\textsf{d}})=0\right]\\&&+\mathbb{P}\left[\mathcal{A}(\tilde{\Phi}^{\textsf{d}})=0\big|\mathcal{C}(\tilde{\Phi}^{\textsf{d}})\neq 0\right]
\mathbb{P}\left[\mathcal{C}(\tilde{\Phi}^{\textsf{d}})\neq 0\right].
\end{eqnarray*}
Since $\tilde{\Phi}^{\textsf{d}}$ is generated from $\Phi^{\textsf{d}}$ and $\mathcal{C}$ covers all clusters of $\mathcal{Z}^{\textsf{d}}$, $\mathbb{P}\left[\mathcal{C}(\tilde{\Phi}^{\textsf{d}})=0\right]=\mathbb{P}[\mathcal{C}(\Phi^{\textsf{d}})=0]$ and thus it follows that
\begin{eqnarray*}
\mathbb{P}\left[\mathcal{A}(\tilde{\Phi}^{\textsf{d}})=0\right]&=& 1\cdot \mathbb{P}[\mathcal{C}(\Phi^{\textsf{d}})=0]+\mathbb{P}\left[\mathcal{A}(\tilde{\Phi}^{\textsf{d}})=0 \big|\mathcal{C}(\Phi^{\textsf{d}})\neq 0\right](1-\mathbb{P}[\mathcal{C}(\Phi^{\textsf{d}})=0]) \\
&=& \exp(-\lambda_{\textsf{d}} \mu(\mathcal{C}))+\mathbb{P}\left[\mathcal{A}(\tilde{\Phi}^{\textsf{d}})=0\big|\mathcal{C}(\Phi^{\textsf{d}})\neq 0\right](1-\exp(-\lambda_{\textsf{d}} \mu(\mathcal{C}))).
\end{eqnarray*}
Now letting $\mu(\mathcal{C})\rightarrow \infty$ (i.e. considering an infinitely large network), we can have $$\mathbb{P}\left[\mathcal{A}(\tilde{\Phi}^{\textsf{d}})=0|\mathcal{C}(\Phi^{\textsf{d}})\neq 0\right]=\lim_{\mu(\mathcal{C})\rightarrow \infty}\left(1-\frac{\lambda_{\textsf{d}}\,\mu(\mathcal{A})}{\lambda_{\textsf{d}}\,\mu(\mathcal{C})}\right)^{\lambda_{\textsf{d}}\mu(\mathcal{C})}=\exp(-\lambda_{\textsf{d}} \mu(\mathcal{A})).$$ Therefore, $\tilde{T}_{\textsf{d}}(\mathcal{A})=1-\exp(-\lambda_{\textsf{d}}\mu(\mathcal{A}))$ and $\tilde{\Phi}^{\textsf{d}}$ is a stationary PPP.
\end{proof}

\section{Proofs of Lemmas and Theorems}
\subsection{Proof of Lemma \ref{Lem:NonCoopConnRxIntenNakaFading}}\label{App:ProofNonCoopConnRxInten}
The Laplace functional of the stationary PPP $\Phi$ for a nonnegative function $g:\mathbb{R}^2\rightarrow\mathbb{R}_+$ is defined and shown as follows \cite{FBBB09}:
\begin{equation*}
 \mathcal{\tilde{L}}_{\Phi}(g)\defn \mathbb{E}\left[e^{-\int_{\mathbb{R}^2}g(X)\,\Phi(\dif X)}\right]=\exp\left(-\int_{\mathbb{R}^2}\left(1-e^{-g(X)}\right)\lambda_{\textsf{r}}\,\mu(\dif X)\right).
\end{equation*}
Since the Laplace functional completely characterizes the distribution of the point process, we can find the intensity of $\Phi^{\textsf{c}}$ by looking for $\mathcal{\tilde{L}}_{\Phi^{\textsf{c}}}(g)$. Recall that $\Phi^{\textsf{c}}=\left\{Y_j\in\Phi^{\textsf{r}}_0: \max_{t\in [1,\cdots,\tau]} H_j(t) \ell(|Y_j|) \geq \beta I_0\right\}$ and $\{H_j(t)\}$ are i.i.d. $\forall t\in [1,2,\cdots,\tau]$. Let $\mathds{1}_\mathcal{A}(x)$ be an indicator function which is equal to 1 if $x\in \mathcal{A}$ and 0, otherwise. The Laplace functional of $\Phi^{\textsf{c}}$ for $g(Y)=\tilde{g}(Y)\mathds{1}_{\Phi^{\textsf{c}}}(Y)$ is given by
\begin{eqnarray*}
\mathcal{\tilde{L}}_{\Phi^{\textsf{c}}}(g)&=& e^{-\pi s^2\lambda_{\textsf{r}}} \sum_{i=0}^{\infty} \frac{\lambda^i_{\textsf{r}}}{i!}\int_{\mathcal{R}_0}\cdots\int_{\mathcal{R}_0} \prod_{j=1}^i \left(e^{-g(Y_j)}\mathbb{P}[Y_j\in\Phi^{\textsf{c}}]+\mathbb{P}[Y_j\notin\Phi^{\textsf{c}}]\right)\mu(\dif Y_1)\cdots \mu(\dif Y_i)\\
&=& e^{-k} \sum_{i=0}^{\infty} \frac{1}{i!}\left(\int_{\mathcal{R}_0}\left(e^{-g(Y)}\mathbb{P}[Y\in\Phi^{\textsf{c}}]+1-\mathbb{P}[Y\in\Phi^{\textsf{c}}]\right)\lambda_{\textsf{r}}\,\mu(\dif Y)\right)^i\\
&=& \exp\left(-\int_{\mathcal{R}_0}\left(1-e^{-g(Y)}\right)\mathbb{P}[Y\in\Phi^{\textsf{c}}]\lambda_{\textsf{r}}\,\mu(\dif Y)\right).
\end{eqnarray*}
Also, for all $r\in[1,s]$ and considering $\lambda_{\textsf{t}}=\Theta(\epsilon)$, we have
\begin{eqnarray}
\mathbb{P}[Y\in\Phi^{\textsf{c}}] &=& \mathbb{P}\left[H_j(t) \ell(|Y|)\geq \beta I_0(t), t=1,\ldots,\tau\right]\nonumber\\
&\stackrel{(a)}{=}& 1-(\mathbb{P}[H\ell(|Y|)<\beta I_0])^{\tau}\nonumber\\
 &=& 1-\left(\mathbb{E}_{I}[F_H(\beta \, r^{\alpha}\,I_0|I_0)]\right)^{\tau}.\label{Eqn:ConnProbNtrans1}
\end{eqnarray}
where $(a)$ follows from the fact that the temporal correlation of interference can be neglected for small $\lambda_{\textsf{t}}$\cite{RKGMH09}. Thus, we have
\begin{equation*}
\mathcal{\tilde{L}}_{\Phi^{\textsf{c}}}(g)= \exp\left(-\int_1^s \left(1-e^{-g(r)}\right)\lambda_{\textsf{c}}(r,\tau)\,\mu(\dif r)\right),
\end{equation*}
where $\lambda_{\textsf{c}}(r,\tau)=\lambda_{\textsf{r}}(1-\left(\mathbb{E}_{I}[F_H(\beta \, r^{\alpha}I_0|I_0)]\right)^{\tau})$. From the above result we know that $\Phi^{\textsf{c}}\subseteq \Phi_0^{\textsf{r}}$ is a nonhomogeneous PPP because its intensity $\lambda_{\textsf{c}}(r,\tau)$ is the intensity $\lambda_{\textsf{r}}$ of $\Phi_0^{\textsf{r}}$ scaled by \eqref{Eqn:ConnProbNtrans1}.

First consider $\sqrt{H}$ is Rayleigh fading. We can have the following:
\begin{eqnarray*}
\mathbb{E}_{I}\left[F_{H}(\beta r^{\alpha}I_0)|I_0)\right] &=& \int_{\mathbb{R}_+} F_H(\beta r^{\alpha}\omega)f_{I_0}(\omega)\, \dif \omega \\
&=& 1-\int_0^{\infty}e^{-\beta r^{\alpha}w}f_{I_0}(w)\,\dif w \stackrel{(b)}{=} 1-\mathcal{L}_{I_0}(\beta r^{\alpha}),
\end{eqnarray*}
where $(b)$ follows from the result of Proposition \ref{Prop:LapalceShotPPP} in Appendix \ref{App:ProofLaplaceShotPPP}. So using \eqref{LaplaceShotPPP1} and \eqref{Eqn:Delta1}, it follows that
\begin{equation*}
\lambda_{c}(r,\tau) = \lambda_{\textsf{r}}\left\{1-\left[1-\exp\left(-\pi \Delta_1(\beta r^{\alpha},\infty) \lambda_{\textsf{t}}\right)\right]^{\tau}\right\}.
\end{equation*}
Now if $\sqrt{H}$ is a Nakagami-\emph{m} random variable, then we can have the following:
\begin{eqnarray*}
\mathbb{E}_{I}\left[F_{H}(\beta r^{\alpha}I_0)|I_0)\right] &=& \int_{\mathbb{R}_+} F_H(\beta r^{\alpha}\omega)f_{I_0}(\omega)\, \dif \omega\\
&=& 1-\int_{0}^{\infty}\frac{\Gamma(m,m\beta r^{\alpha}\omega)}{\Gamma(m)}f_{I_0}(\omega)\,\dif \omega,
\end{eqnarray*}
where $\Gamma(y,x)=\int_{x}^{\infty} t^{y-1}e^{-t}\,\dif t$ is the incomplete Gamma function and $\Gamma(m,0)=\Gamma(m)=(m-1)!$. Also, if $f_W(w)$ is a probability density function of random variable $W$ then we can have the following result for $m>1$:
\begin{eqnarray}
\int_{0}^{\infty} \Gamma(m,a w)\,f_W(w)\,\dif w &=& \int_{0}^{\infty}\int_{a w}^{\infty} t^{m-1}e^{-t}\, f_W(w)\,\dif t\, \dif w \nonumber\\
&=& (-1)^{m-1}\,a^m\,\frac{\dif ^{m-1}}{\dif a^{m-1}}\left(\frac{\mathcal{L}_{W}(a)}{a}\right).
\end{eqnarray}
According to Proposition \ref{Prop:LapalceShotPPP} in Appendix \ref{App:ProofLaplaceShotPPP}, we can have $\mathcal{L}_{I_0}(\phi)$ with $\Delta_1(\phi,\infty)$. Thus,
\begin{eqnarray}\label{Eqn:SuccProbNakaFading1}
\mathbb{E}_{I}\left[F_{H}(\beta r^{\alpha}I_0|I_0))\right] &=& 1+\frac{(-\phi)^m}{\Gamma(m)}\frac{\dif^{m-1}}{\dif \phi^{m-1}}\left(\frac{\mathcal{L}_{I_0}(\phi)}{\phi}\right)\bigg|_{\phi=m\beta r^{\alpha}}\nonumber \\
&=& 1-\Psi^{(m-1)}(m\beta r^{\alpha}).
\end{eqnarray}
Substituting \eqref{Eqn:SuccProbNakaFading1} into \eqref{Eqn:ConnProbNtrans1}, then the result in \eqref{Eqn:NonCoopConnIntenNakaFading} can be arrived.

\subsection{Proof of Theorem \ref{Thm:MaxContenIntenNonCoop}}\label{App:ProofMaxContenIntenNonCoop}
Here we only provide the proof for Nakagami-$m$ faidn with $m>1$ since the proof for Rayleigh fading is similar. According to the outage probability \eqref{Eqn:DefOutProbMTC2} upper bounded by $\epsilon$, we know
$$\mathbb{E}_{R}[\lambda_{\textsf{c}}(R,\tau)]\geq \lambda_{\textsf{r}}+\frac{\log(1-\epsilon)}{\pi s^2}=\lambda_{\textsf{r}}\left(1-\frac{\epsilon}{k}\right)+\Theta(\epsilon^2),$$
for sufficiently small $\epsilon$. Using the upper bound in \eqref{Eqn:AvgConnIntenNakagami} and the lower bound on $\mathbb{E}_{R}[\lambda_{\textsf{c}}(R,\tau)]$ obtained in above, it follows that
\begin{equation}\label{Eqn:AvgConnectProb1}
\mathbb{E}_{R}\left[(1-\Psi^{(m-1)}(m\beta\,R^{\alpha}))\right]^{\tau}\leq \frac{\epsilon}{k}.
\end{equation}
Therefore, we further have
\begin{equation}
\mathbb{E}_R[\Psi^{(m-1)}(m\beta R^{\alpha})]\geq 1-\sqrt[\tau]{\frac{\epsilon}{k}}.
\end{equation}

Note that $\mathbb{E}_{R}\left[\Psi^{(m-1)}(m\beta\,R^{\alpha})\right]$ is upper bounded by one and thus it approaches to unity when $\epsilon/k$ is sufficiently small such that $\sqrt[\tau]{\epsilon/k}\leq \epsilon$. That means $\Psi^{(m-1)}(m\beta\,R^{\alpha})$ is very close to 1 almost surely and thus $\lambda_{\textsf{t}}=\Theta(\epsilon)$. If $k$ is sufficiently large, then we have $\exp(-\pi\,\lambda_{\textsf{t}}\Delta_1)=1-\pi\,\lambda_{\textsf{t}}\Delta_1+\Theta(\epsilon^2)$. Substituting this expression into \eqref{Eqn:PsiMphi}, $\Psi^{(m-1)}(m\beta R^{\alpha})$ can be reduced to $1- \pi\,\lambda_{\textsf{t}}\,\Delta_1\,\prod_{j=1}^{m-1}(1-2/j\alpha)+\Theta(\epsilon^2)$.
Define $\hat{\Delta}_1(\phi)\defn [\Delta_1(\phi,\infty)-\Delta_1(\phi,\hat{a})]\prod_{j=1}^{m-1}(1-2/j\alpha)$. Choosing $0\leq\hat{a}\leq 1$, we have
\begin{eqnarray*}
\mathbb{E}_{R}\left[(1-\Psi^{(m-1)}(m\beta\,R^{\alpha}))^{\tau}\right] &\leq& \left[\hat{\Delta}_1(\beta)\,\pi\,\lambda_{\textsf{t}}\right]^{\tau}\mathbb{E}[R^{2\tau}]+\Theta(\epsilon^2)\\
&=& \frac{\left[\pi s^2\hat{\Delta}_1(\beta)\,\lambda_{\textsf{t}}\right]^{\tau}}{(\tau+1)}+ \Theta(\epsilon^2).
\end{eqnarray*}
According to \eqref{Eqn:AvgConnectProb1} and the above result, these two upper bounds should coincide when the maximum contention intensity $\bar{\lambda}_{\textsf{t}}$ is reached. Namely, \eqref{Eqn:MaxContenIntenNoCoop} is obtained from the following:
\begin{equation}\label{Eqn:MaxIntensityDenseNwks}
\bar{\lambda}_{\textsf{t}} = \frac{\sqrt[\tau]{\epsilon\,(\tau+1)}}{\pi s^2\,\beta^{\frac{2}{\alpha}}\,\sqrt[\tau]{k}\,\hat{\Delta}_1(\beta)}+\Theta(\epsilon^2) \Rightarrow \bar{\lambda}_{\epsilon} = \eta \cdot \frac{\sqrt[\tau]{\epsilon\,(\tau+1)}}{\pi s^2\,\beta^{\frac{2}{\alpha}}\,\sqrt[\tau]{k}},
\end{equation}
where $\eta \defn 1/\hat{\Delta}_1(\beta)$. This completes the proof.

\subsection{Proof of Lemma \ref{Lem:NonCoopBoundsMCrate}}\label{App:ProofBoundsMCrate}
According to Campbell's theorem, we know $\mathbb{E}[I_0]=\frac{2\pi \lambda_{\textsf{t}}}{\alpha-2}$ and then the lower bound on $b$ can be reduced as follows.
\begin{eqnarray*}
b \stackrel{(a)}{\geq} \mathbb{E}\left[\log \left(1+\frac{H_{\max}}{s^{\alpha}\mathbb{E}[I_0]}\right)\right]
  &\stackrel{(b)}{\geq}& \mathbb{E}\left[\log \left(1+\frac{H_{\max}(\alpha-2)}{2\pi s^2\lambda_{\textsf{t}}}\right)\right]\\
  &\geq& \log\left(1+\frac{1}{\pi s^2\lambda_{\textsf{t}}}\right)F^{\texttt{c}}_{H_{\max}}\left(\frac{2}{\alpha-2}\right),
\end{eqnarray*}
where $(a)$ follows from Jensen's inequality since $\log(1+a/x)$ is convex for $x>0$ and constant $a>0$, $(b)$ is due to $s^{\alpha-2}\geq 1$, and $F^{\texttt{c}}_{H_{\max}}$ is the CCDF of random variable $H_{\max}$. In addition, we also know $b$ in \eqref{Eqn:MaxBcRateWoTimeDiv} is upper bounded as follows.
\begin{equation*}
b \leq \log \left(\mathbb{E}[I_0]+\mathbb{E}[H_{\max}]s^{-\alpha}\right)+\mathbb{E}\left[\log\left(1/I_0\right)\right].
\end{equation*}
The upper bound is obtained by first conditioning on $I_0$ and then using Jensen's inequality. We know the term $\log(1/I_0)$ in the upper bound is convex and thus $\mathbb{E}[\log(1/I_0)]\geq \log(1/\mathbb{E}[I_0])$. Since the network is interference-limited, $\mathbb{E}[I_0]$ must be bounded above zero. Thus there exists a $\gamma_1>0$ such that $\log(\gamma_1) \leq \mathbb{E}[\log(I_0)] \leq \log(\mathbb{E}[I_0])$, which means $\mathbb{E}[\log(I_0)]\geq \log(\gamma_2\,\mathbb{E}[I_0])$ for any $\gamma_2\in(0,\gamma_1/\mathbb{E}[I_0]]$. So the upper bound can be simplified as follows:
\begin{eqnarray*}
\log \left(\mathbb{E}[I_0]+\mathbb{E}[H_{\max}]s^{-\alpha}\right)+\mathbb{E}\left[\log\left(\frac{1}{I_0}\right)\right]
&\leq& \log\left(\frac{\mathbb{E}[I_0]}{\mathbb{E}[H_{\max}]}+\mathbb{E}[R^{-\alpha}]\right)
+\log\left(\frac{\mathbb{E}[H_{\max}]}{\gamma_2\mathbb{E}[I_0]}\right)\\ &\leq&\log\left(1+\frac{1}{\pi s^2\lambda_{\textsf{t}}}\right)+\log\left(\frac{\mathbb{E}[H_{\max}]}{\gamma_2}\right)
\end{eqnarray*}
because $s^{-\alpha}\leq \mathbb{E}[R^{-\alpha}]\leq\frac{2}{s^2(\alpha-2)}$ and $\mathbb{E}[H_{\max}]\geq 1$. Note that the CDF of $H_{\max}$, $F_{H_{\max}}$, is equal to $(F_H)^{\tau}$ since $\{H_j(t), t=1,\ldots,\tau\}$ are i.i.d. random variables. The proof is complete.

\subsection{Proof of Theorem \ref{Thm:MultihopMulticastMaxInten}}\label{App:ProofSectorMaxContenInten}
We start with the derivation of the upper bound on $\mathbb{E}[\lambda_{\textsf{c}}]$ as follows.
\textbf{(i)} \textbf{The upper bound on} $\mathbb{E}[\lambda_{\textsf{c}}]$: Suppose the multicast cluster $\mathcal{R}_0$ is tessellated into $v$ smaller multicast regions of equal area $\pi s^2/v$. Let $\mathcal{T}$ denote the set of all the pathes from the source multicast region (\ie the tessellated region in which the central transmitter is located) to the last multicast region. Let $T_*\in\mathcal{T}$ be the chosen path for multicasting a packet and $\mathcal{S}_i$ denote the $i$th tessellated region on $T_*$. Since there are $v$ multicast regions on $T_*$, the duration of a time slot for multicasting a packet in each region is at most $\tau/v$ transmission attempts. In order to implement multihop multicast, the multicast outage event should be redefined here as follows. We say there exists a multihop multicast outage in a multicast time slot if any of the receivers in the current multicast region does not receive the packet or no receivers in the next multicast region receive the packet after $\tau/v$ transmission attempts. In addition, $T_*$ is chosen according to the following criteria:
\begin{equation}\label{Eqn:OptMultiPath}
 T_* = \arg \min_{T\in\mathcal{T}} \left[1-\prod_{i=1}^{v} \left(1-\mathbb{P}[\mathcal{E}_i|T]\right)\right],
\end{equation}
where $\mathcal{E}_i$ is the multicast outage event happening in $\mathcal{S}_i$. So $T_*$ is the path with the minimum probability of end-to-end multicast outage.

If $\Phi_{s_i}$ is the receiver-connected process in $\mathcal{S}_i$ after $\tau$ attempts, then $\mathcal{E}_i$ can be expressed as $\{\Phi_{s_i}\subset(\Phi_r\cap\mathcal{S}_i)\cup (\tilde{Y}_{i+1}\in \mathcal{S}_{i+1})\}$ where $\tilde{Y}_{i+1}$ is any intended receiver in $\mathcal{S}_{i+1}$ and it could be the transmitter in the $(i+1)$th time slot. Due to the broadcast nature of wireless communication, the intended receivers in the regions other than $\mathcal{S}_i$ could also receive the packet which is merely multicasted to $\mathcal{S}_i$ in the $i$th time slot. Thus, if the average intensity of $\Phi_{s_i}$ at the end of the $t$th time slot is $\mathbb{E}_X[\lambda_{s_i}(X,t\tau/v)]$ where $X\in\mathcal{R}_0$, then we know $\mathbb{E}_X[\lambda_{s_1}(X,t\tau/v)]\leq \mathbb{E}_X[\lambda_{s_2}(X,t\tau/v)]\leq \cdots \leq \mathbb{E}_X[\lambda_{s_{v}}(X,t\tau/v)]$ since $\Phi_{s_i}$ is a filtration process.

Since the Lebesgue measure of $\mathcal{S}_i$ is $\frac{\pi s^2}{v}$, the probability of $E_i$ on $T_*$ is given by
\begin{eqnarray*}
 \mathbb{P}[\mathcal{E}_i|T_*] &=& 1-\exp\left(-\int_{\mathcal{S}_i}(\lambda_{\textsf{r}}-\lambda_{s_i}(X,i\tau/v))\, \mu(\dif X)\right)\\
 &&\cdot \left[1-\exp\left(-\frac{\pi s^2\mathbb{E}_X[\lambda_{s_{i+1}}(X,i\tau/v)]}{v}\right)\right]\leq \epsilon_v,
\end{eqnarray*}
where $\epsilon_v=1-(1-\epsilon)^{1/v}$ is the upper bound constraint on the multicast outage probability in each tessellated region. Considering sufficiently large $k$, the above equation can be simplified as
\begin{equation*}
\mathbb{P}[\mathcal{E}_i|T_*]=1-\exp\left(-\int_{\mathcal{S}_i}(\lambda_{\textsf{r}}-\lambda_{s_i}(X,i\tau/v))\, \mu(\dif X)\right)+\Theta(\epsilon^2).
\end{equation*}
The multicast outage event in $\mathcal{R}_0$ is the union of all the multihop multicast outage events, \ie $\mathcal{E}=\bigcup_{i=1}^{v}\mathcal{E}_i=(\bigcap_{i=1}^{v}\mathcal{E}_i^\texttt{c})^\texttt{c}$. Since all multihop multicast outage events are independent, the probability of $\mathcal{E}$ can be explicitly expressed as
\begin{eqnarray*}
\mathbb{P}[\mathcal{E}|T_*] &=& 1-\prod_{i=1}^{v}\mathbb{P}\left[\mathcal{E}_i^\texttt{c}|T_*\right]\\
&=& 1-\exp\left\{-\pi s^2\lambda_{\textsf{r}}+\sum_{i=1}^{v}\int_{\mathcal{S}_i}
\lambda_{s_i}(X,i\tau/v)\, \mu(\dif X)\right\}+\Theta(\epsilon^2)\leq \epsilon.
\end{eqnarray*}
According to \eqref{Eqn:DefOutProbMTC2} and the above equation, it follows that
\begin{equation*}
 \mathbb{E}_X[\lambda_{\textsf{c}}(X,\tau)]=\frac{1}{v}\sum_{i=1}^{v}\int_{\mathcal{S}_i}
\lambda_{s_i}(X,i\tau/v)\,\frac{ \mu(\dif X)}{\pi s^2/v}=\frac{1}{v} \sum_{i=1}^{v}\mathbb{E}_{X}[\lambda_{s_i}(X,i\tau/v)].
\end{equation*}
From the results in Lemma \ref{Lem:BoundsAveConnInten}, we can know $\mathbb{E}_X[\lambda_{s_i}(X,i\tau/v)]$ with Rayleigh and Nakagami-$m$ fading at the $i$th time slot. Let $A(r)=1-\exp(-\pi\Delta_1(\beta r^{\alpha},\infty)\lambda_{\textsf{t}})$ for Rayleigh fading or $A(r)=1-\Psi^{(m-1)}(m\beta r^{\alpha})$ for Nakagami-$m$ fading with $m>1$, and thus the upper bound on $\mathbb{E}_R[\lambda_{\textsf{c}}(R,\tau)]$ can be shown as
\begin{eqnarray*}
 \mathbb{E}_R[\lambda_{\textsf{c}}(R,\tau)]\leq \lambda_{\textsf{r}}\left(1-\frac{1}{v}\sum_{i=1}^{v}\mathbb{E}_R\left[A(R)\right]^{\frac{i\tau}{v}}\right).
\end{eqnarray*}

\textbf{(ii)} \textbf{The scaling law of $\bar{\lambda}_{\epsilon}$}: Using the upper bound on $\mathbb{E}_R[\lambda_{\textsf{c}}(R,\tau)]$ obtained From \eqref{Eqn:DefOutProbMTC2} and the above result, it yields the following inequality:
\begin{equation*}
    \frac{\epsilon\,v}{k}\geq \sum_{i=1}^{v}\mathbb{E}_R[A(R)]^{\frac{i\tau}{v}}=\frac{\mathbb{E}_R[A]^{\frac{\tau}{v}}(1-\mathbb{E}_R[A]^{\tau})}{1-\mathbb{E}_R[A]^{\frac{\tau}{v}}}.
\end{equation*}
Also, we know $\epsilon_v=1-\sqrt[v]{1-\epsilon}=\frac{\epsilon}{v}+\Theta(\epsilon^2)$ and $\mathbb{E}_R[A]^{\frac{\tau}{v}}\leq \frac{\epsilon_v\,v}{k}=\frac{\epsilon}{k}+\Theta(\epsilon^2)$.
Thus, we know
\begin{equation*}
    1-\left(\frac{\epsilon\,v}{k}\right)^{\frac{v}{\tau}}\leq \mathbb{E}_R [\Psi^{(m-1)}(m\beta R^{\alpha})].
\end{equation*}
 $\mathbb{E}_R [\Psi^{(m-1)}(m\beta R^{\alpha})]\approx 1$ since $\mathbb{E}_R [\Psi^{(m-1)}(m\beta R^{\alpha})]\leq 1$ and $k\geq \frac{v}{\epsilon^{\tau/v-1}}$. So the random variable $\Psi^{(m-1)}(m\beta R^{\alpha})$ is very close to one almost surely and thus $\lambda_{\textsf{t}}$ is $\Theta(\epsilon)$. Then following the same steps in the proof of Theorem \ref{Thm:MaxContenIntenNonCoop}, we can show that
\begin{equation}
    \bar{\lambda}_{\epsilon} = \eta\cdot\frac{v^{\frac{v}{\tau}+1}\,\tau^2\,\rho}{\pi s^2 \beta^{\frac{2}{\alpha}}k^{\frac{v}{\tau}}}=\Theta\left(\frac{v^{\frac{v}{\tau}+1}\,\tau^2\,\rho}{\pi s^2 \beta^{\frac{2}{\alpha}}k^{\frac{v}{\tau}}}\right).
\end{equation}
Thus, \eqref{Eqn:MultihopMaxContenInten} is arrived, which completes the proof.

\bibliographystyle{ieeetran}
\bibliography{IEEEabrv,Ref_MultiCastTC}

\newpage

\begin{figure}[!h]
  \centering
  \includegraphics[scale=0.75]{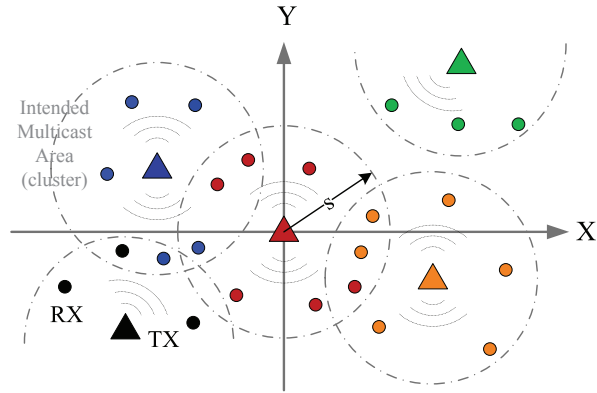}\\
  \caption{Multicast transmission model in a planar wireless ad hoc network: The transmitters in the network form a stationary PPP of intensity $\lambda_{\textsf{t}}$. Each transmitter (triangle) has an intended multicast area of radius $s$, where all the intended receivers (small circles) also form a stationary PPP of intensity $\lambda_{\textsf{r}}$. A transmitter and its corresponding intended receivers are indicated by the same color in a cluster. So each cluster could contain other transmitters and unintended receivers in addition to its own transmitter and intended receivers.}
  \label{Fig:MulticastModel}
\end{figure}


\begin{figure}[!h]
\centering
  \includegraphics[scale=0.65]{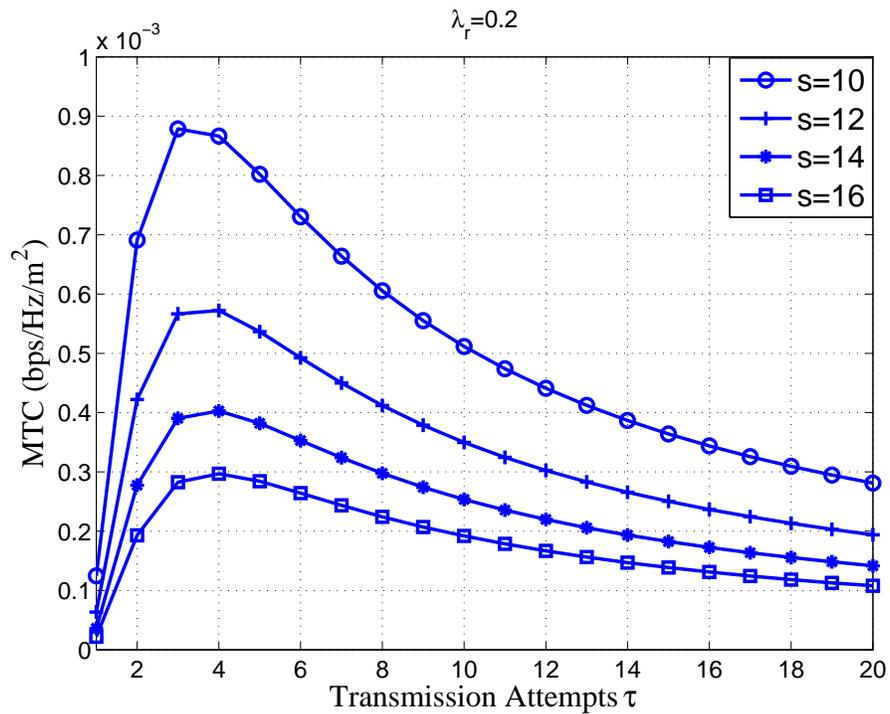}\\
  \caption{The simulated MTC of a large network with Rayleigh fading for $\epsilon=0.1$, $\beta=2$, $\alpha=3$ and $\lambda_r=0.1$.}\label{Fig:ShMTCwoRxCoop1}
\end{figure}


\begin{figure}[!h]
\centering
  \includegraphics[scale=0.65]{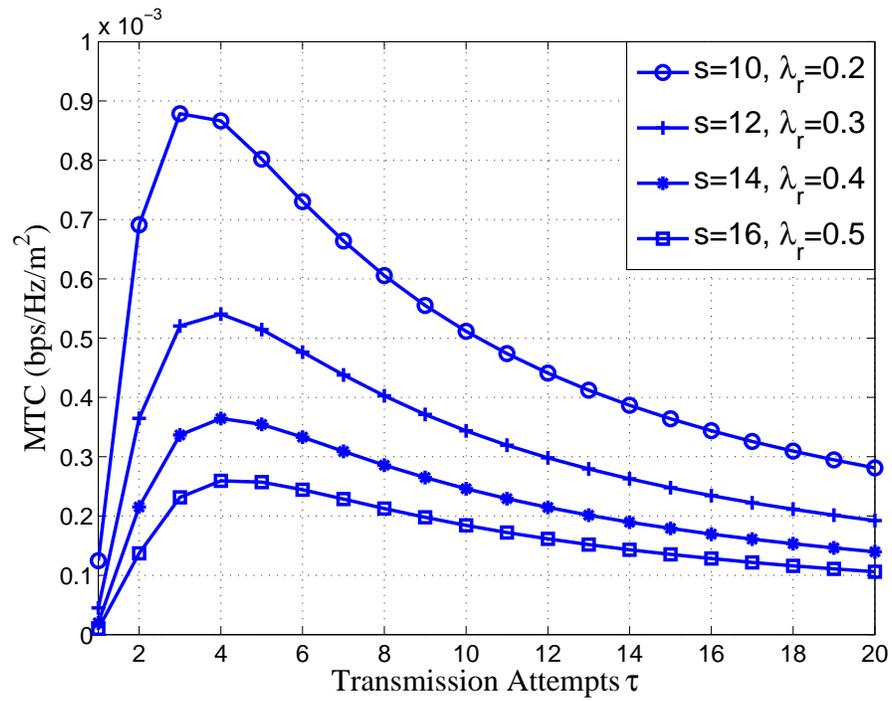}\\
  \caption{The simulated MTC of a large dense network with Rayleigh fading for $\epsilon=0.1$, $\beta=2$ and $\alpha=3$.}\label{Fig:ShMTCwoRxCoop2}
\end{figure}




\begin{figure}[!h]
  \centering
  \includegraphics[scale=0.75]{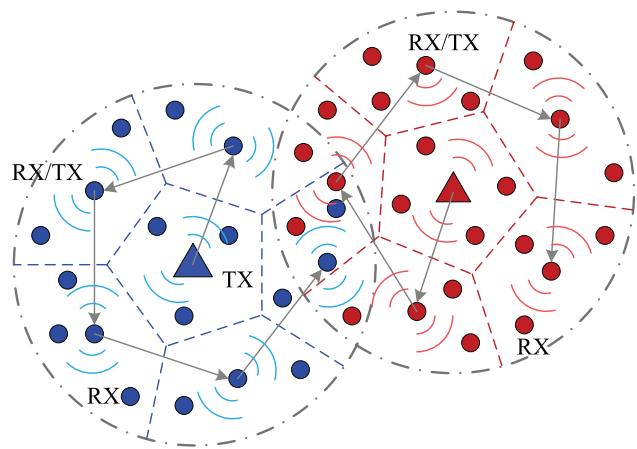}\\
  \caption{Each multicast cluster in a planar network is tesselated into 6 smaller multicast regions of equal area. The arrows in each cluster show an example path of delivering a packet over the tessellated regions by the multihop multicast method. }
  \label{Fig:MultihopTransModel}
\end{figure}


\begin{figure}[!h]
\centering
  \includegraphics[scale=0.6]{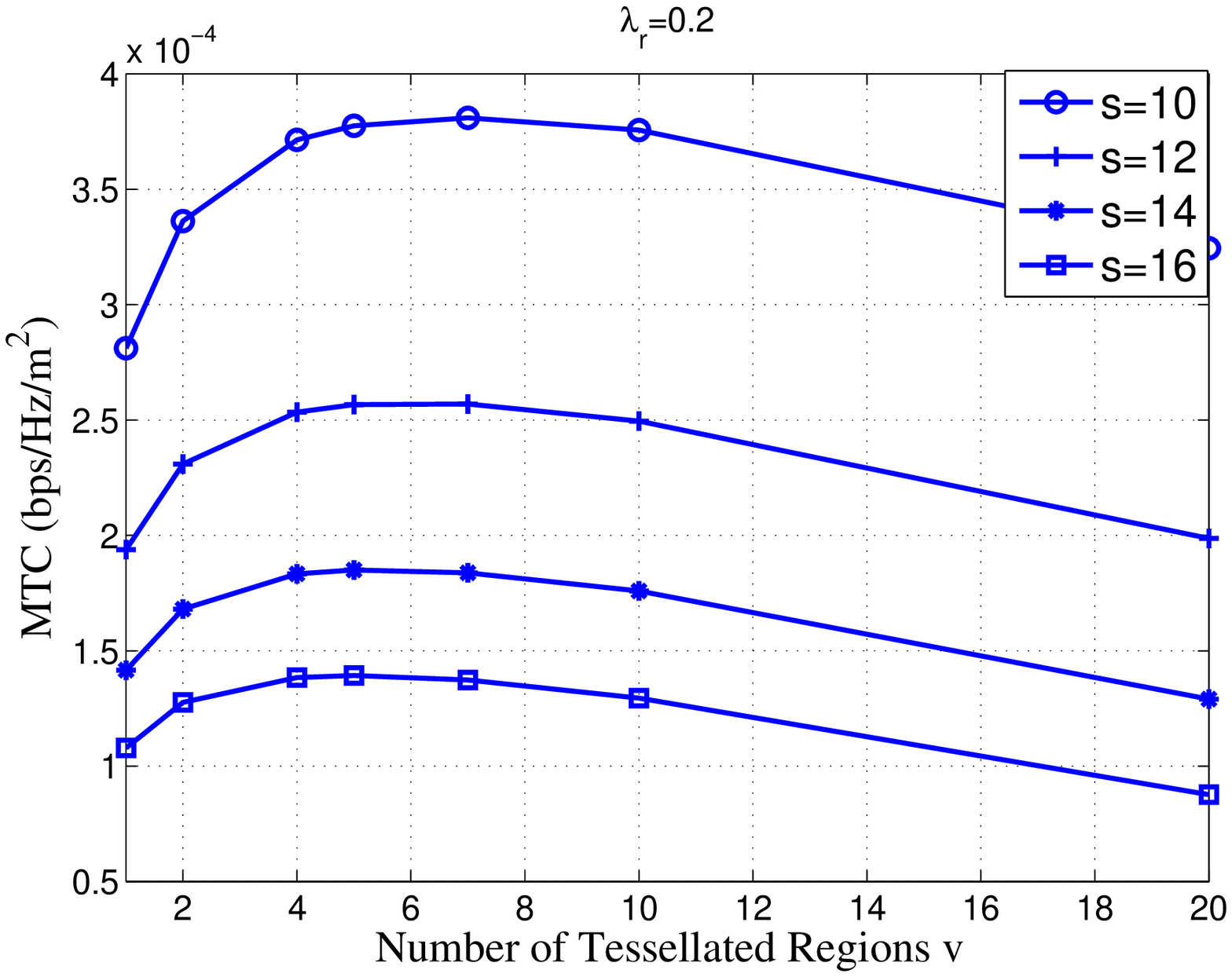}\\
  \caption{The simulated MTC achieved by multihop multicast in a large network with Rayleigh fading for $\epsilon=0.1$, $\beta=2$, $\alpha=3$, $\tau=20$ and $\lambda_{\textsf{r}}=0.2$.}\label{Fig:MhMTCwRxCoop1}
\end{figure}


\begin{figure}[!h]
\centering
  \includegraphics[scale=0.6]{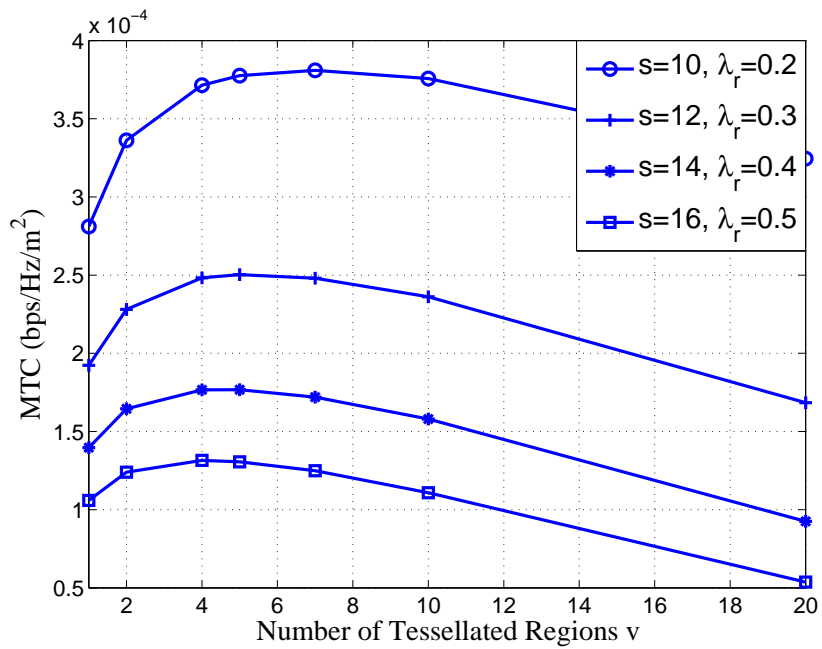}\\
  \caption{The simulated MTC achieved by multihop multicast in a large dense network with Rayleigh fading for $\epsilon=0.1$, $\beta=2$, $\alpha=3$ and $\tau=20$.}\label{Fig:MhMTCwRxCoop2}
\end{figure}

\end{document}